\documentclass{article}

\def\noheaderplainsetup{

\topmargin=0pt \headheight=0pt \headsep=0pt  \oddsidemargin=0pt \evensidemargin=0pt  \textheight=8.9truein \textwidth=6.2truein}

\noheaderplainsetup

\usepackage{amsfonts}

\begin{document}


\newcommand{\code}[1]{\ulcorner #1 \urcorner}
\newcommand{\mldi}{\hspace{2pt}\mbox{\footnotesize $\vee$}\hspace{2pt}}
\newcommand{\mlci}{\hspace{2pt}\mbox{\footnotesize $\wedge$}\hspace{2pt}}
\newcommand{\emptyrun}{\langle\rangle} 
\newcommand{\oo}{\bot}            
\newcommand{\pp}{\top}            
\newcommand{\xx}{\wp}               
\newcommand{\legal}[2]{\mbox{\bf Lr}^{#1}_{#2}} 
\newcommand{\win}[2]{\mbox{\bf Wn}^{#1}_{#2}} 
 \newcommand{\one}{\mbox{\sc One}}
 \newcommand{\two}{\mbox{\sc Two}}
 \newcommand{\three}{\mbox{\sc Three}}
 \newcommand{\four}{\mbox{\sc Four}}
 \newcommand{\first}{\mbox{\sc Derivation}}
 \newcommand{\second}{\mbox{\sc Second}}
 \newcommand{\uorigin}{\mbox{\sc Org}}
 \newcommand{\image}{\mbox{\sc Img}}
 \newcommand{\limitset}{\mbox{\sc Lim}}
 \newcommand{\fif}{\mbox{\bf CL15}}
\newcommand{\col}[1]{\mbox{$#1$:}}

\newcommand{\sti}{\mbox{\raisebox{-0.02cm}
{\scriptsize $\circ$}\hspace{-0.121cm}\raisebox{0.08cm}{\tiny $.$}\hspace{-0.079cm}\raisebox{0.10cm}
{\tiny $.$}\hspace{-0.079cm}\raisebox{0.12cm}{\tiny $.$}\hspace{-0.085cm}\raisebox{0.14cm}
{\tiny $.$}\hspace{-0.079cm}\raisebox{0.16cm}{\tiny $.$}\hspace{1pt}}}
\newcommand{\costi}{\mbox{\raisebox{0.08cm}
{\scriptsize $\circ$}\hspace{-0.121cm}\raisebox{-0.01cm}{\tiny $.$}\hspace{-0.079cm}\raisebox{0.01cm}
{\tiny $.$}\hspace{-0.079cm}\raisebox{0.03cm}{\tiny $.$}\hspace{-0.085cm}\raisebox{0.05cm}
{\tiny $.$}\hspace{-0.079cm}\raisebox{0.07cm}{\tiny $.$}\hspace{1pt}}}

\newcommand{\seq}[1]{\langle #1 \rangle}           

\newcommand{\pstb}{\mbox{\raisebox{-0.01cm}{\large $\wedge$}\hspace{-5pt}\raisebox{0.26cm}{\small $\mid$}\hspace{4pt}}}
\newcommand{\pcostb}{\mbox{\raisebox{0.22cm}{\large $\vee$}\hspace{-5pt}\raisebox{0.02cm}{\footnotesize $\mid$}\hspace{4pt}}}

\newcommand{\sst}{\mbox{\raisebox{-0.07cm}{\scriptsize $-$}\hspace{-0.2cm}$\pst$}} 

\newcommand{\scost}{\mbox{\raisebox{0.20cm}{\scriptsize $-$}\hspace{-0.2cm}$\pcost$}} 

\newcommand{\uvalid}{\mbox{$\vdash\hspace{-5pt}\vdash\hspace{-5pt}\vdash$}}


\newcommand{\mla}{\mbox{{\Large $\wedge$}}}
\newcommand{\mle}{\mbox{{\Large $\vee$}}}

\newcommand{\pst}{\mbox{\raisebox{-0.01cm}{\scriptsize $\wedge$}\hspace{-4pt}\raisebox{0.16cm}{\tiny $\mid$}\hspace{2pt}}}
\newcommand{\gneg}{\neg}                  
\newcommand{\mli}{\rightarrow}                     
\newcommand{\cla}{\mbox{\large $\forall$}}      
\newcommand{\cle}{\mbox{\large $\exists$}}        
\newcommand{\mld}{\vee}    
\newcommand{\mlc}{\wedge}  
\newcommand{\ade}{\mbox{\Large $\sqcup$}}      
\newcommand{\ada}{\mbox{\Large $\sqcap$}}      
\newcommand{\add}{\sqcup}                      
\newcommand{\adc}{\sqcap}                      

\newcommand{\tlg}{\bot}               
\newcommand{\twg}{\top}               
\newcommand{\st}{\mbox{\raisebox{-0.05cm}{$\circ$}\hspace{-0.13cm}\raisebox{0.16cm}{\tiny $\mid$}\hspace{2pt}}}
\newcommand{\cst}{{\mbox{\raisebox{-0.05cm}{$\circ$}\hspace{-0.13cm}\raisebox{0.16cm}{\tiny $\mid$}\hspace{1pt}}}^{\aleph_0}} 
\newcommand{\cost}{\mbox{\raisebox{0.12cm}{$\circ$}\hspace{-0.13cm}\raisebox{0.02cm}{\tiny $\mid$}\hspace{2pt}}}
\newcommand{\ccost}{{\mbox{\raisebox{0.12cm}{$\circ$}\hspace{-0.13cm}\raisebox{0.02cm}{\tiny $\mid$}\hspace{1pt}}}^{\aleph_0}} 
\newcommand{\pcost}{\mbox{\raisebox{0.12cm}{\scriptsize $\vee$}\hspace{-4pt}\raisebox{0.02cm}{\tiny $\mid$}\hspace{2pt}}}

\newcommand{\psti}{\mbox{\raisebox{-0.02cm}{\tiny $\wedge$}\hspace{-0.121cm}\raisebox{0.08cm}{\tiny $.$}\hspace{-0.079cm}\raisebox{0.10cm}
{\tiny $.$}\hspace{-0.079cm}\raisebox{0.12cm}{\tiny $.$}\hspace{-0.085cm}\raisebox{0.14cm}
{\tiny $.$}\hspace{-0.079cm}\raisebox{0.16cm}{\tiny $.$}\hspace{1pt}}}

\newcommand{\pcosti}{\mbox{\raisebox{0.08cm}{\tiny $\vee$}\hspace{-0.121cm}\raisebox{-0.01cm}{\tiny $.$}\hspace{-0.079cm}\raisebox{0.01cm}
{\tiny $.$}\hspace{-0.079cm}\raisebox{0.03cm}{\tiny $.$}\hspace{-0.085cm}\raisebox{0.05cm}
{\tiny $.$}\hspace{-0.079cm}\raisebox{0.07cm}{\tiny $.$}\hspace{1pt}}}


\newtheorem{theoremm}{Theorem}[section]
\newtheorem{conditionss}{Condition}[section]
\newtheorem{thesiss}[theoremm]{Thesis}
\newtheorem{definitionn}[theoremm]{Definition}
\newtheorem{lemmaa}[theoremm]{Lemma}
\newtheorem{notationn}[theoremm]{Notation}\newtheorem{corollary}[theoremm]{Corollary}
\newtheorem{propositionn}[theoremm]{Proposition}
\newtheorem{conventionn}[theoremm]{Convention}
\newtheorem{examplee}[theoremm]{Example}
\newtheorem{remarkk}[theoremm]{Remark}
\newtheorem{factt}[theoremm]{Fact}
\newtheorem{exercisee}[theoremm]{Exercise}
\newtheorem{questionn}[theoremm]{Open Problem}
\newtheorem{conjecturee}[theoremm]{Conjecture}

\newenvironment{exercise}{\begin{exercisee} \em}{ \end{exercisee}}
\newenvironment{definition}{\begin{definitionn} \em}{ \end{definitionn}}
\newenvironment{theorem}{\begin{theoremm}}{\end{theoremm}}
\newenvironment{lemma}{\begin{lemmaa}}{\end{lemmaa}}
\newenvironment{proposition}{\begin{propositionn} }{\end{propositionn}}
\newenvironment{convention}{\begin{conventionn} \em}{\end{conventionn}}
\newenvironment{remark}{\begin{remarkk} \em}{\end{remarkk}}
\newenvironment{proof}{ {\bf Proof.} }{\  \rule{2.5mm}{2.5mm} \vspace{.2in} }
\newenvironment{idea}{ {\bf Idea.} }{\  \rule{1.5mm}{1.5mm} \vspace{.15in} }
\newenvironment{example}{\begin{examplee} \em}{\end{examplee}}
\newenvironment{fact}{\begin{factt}}{\end{factt}}
\newenvironment{notation}{\begin{notationn} \em}{\end{notationn}}
\newenvironment{conditions}{\begin{conditionss} \em}{\end{conditionss}}
\newenvironment{question}{\begin{questionn}}{\end{questionn}}
\newenvironment{conjecture}{\begin{conjecturee}}{\end{conjecturee}}

\title{The countable versus uncountable branching recurrences in computability logic\thanks{Supported by NNSF (60974082) of China.}}
\author{Wenyan Xu \ and \ Sanyang Liu}
\date{}
\maketitle

\begin{abstract} This paper introduces a new simplified version of the countable branching recurrence $\cst$ of Computability Logic, proves its equivalence to the old one, and shows that the basic logic induced by $\cst$ (i.e., the one in the signature $\{\neg,\wedge,\vee,\cst,\ccost\}$) is a proper superset of the basic logic induced by the uncountable branching recurrence $\st$ (i.e., the one in the signature $\{\neg,\wedge,\vee,\st,\cost\}$). A further result of this paper is showing that $\cst$ is strictly weaker than $\st$ in the sense that $\st F$ logically implies $\cst F$ but not vice versa.

\end{abstract}


\



\section{Introduction}\label{ssintr}

{\em Computability logic} (CoL), introduced by G. Japaridze \cite{Jap03,Japfin}, is a formal theory of interactive computational problems, understood as games between a machine and its environment (symbolically named as $\top$ and $\bot$, respectively). Formulas in it represent such problems, logical operators stand for operations on them, and ``truth" means existence of an algorithmic solution, i.e. $\top$'s effective winning strategy.

Among the most important operators of CoL are {\em recurrence operators}, in their overall logical spirit reminiscent of the exponentials of linear logic. Recurrences, in turn, come in several flavors, two most natural and basic sorts of which are {\em countable branching recurrence} $\cst$ and {\em uncountable branching recurrence} $\st$, together with their duals $\ccost,\cost$  defined by $\ccost F=\neg\cst\neg F$ and $\cost F=\neg\st\neg F$. Intuitive discussions and elaborations on the two sorts of recurrences and the relations between them were given in \cite{Japfour,Japsep,Mosc}. However, finding syntactic characterizations  of the logic induced by recurrences had remained among the greatest challenges in CoL until the recent work \cite{Japtam,Japtam2}, where a sound and complete axiomatization, called {\bf CL15},  for the basic $(\neg,\wedge,\vee,\st,\cost)-$fragment of computability logic was constructed.\footnote{The soundness part was proven in \cite{Japtam}, and the completeness part in \cite{Japtam2}.} At the same time, the logical behavior of countable branching recurrence $\cst$ still remains largely ununderstood. It is not even known whether the set of principles validated by $\cst$ is recursively enumerable. The present paper brings some initial light into this otherwise almost completely dark picture. It introduces a new simplified definition of $\cst$ and proves that the new version of $\cst$ is logically equivalent to the old one originally introduced in \cite{Japfour}. Relying on this equivalence, the paper then shows that the set of principles validated by $\cst,\ccost$ in combination with the basic operations $\neg,\wedge,\vee$ is a proper superset of the set of those validated by $\st,\cost$. This is achieved by positively settling Conjecture 6.4 of \cite{Japtam}, according to which {\bf CL15} continues to be sound---but not complete--- with $\cst$ and $\ccost$ instead of $\st$ and $\cost$. Further, to make our investigation of
the relationship between $\cst$ and $\st$ more complete, at the end of this paper we also prove that $\cst$ is strictly weaker than $\st$ in the sense that  $\st F$ logically implies $\cst F$ but not vice versa.


{\bf CL15} is a system built in {\em cirquent calculus}. The latter is a
refinement of sequent calculus. Unlike the more traditional proof theories that manipulate tree-like objects (formulas, sequents, hypersequents, etc.), cirquent calculus deals with graph-style structures called {\em cirquents} (the term is a combination of ``CIRcuit'' and ``seQUENT''), with its main characteristic feature being allowing to explicitly account for possible {\em sharing} of subcomponents between different subcomponents. The approach was introduced by Japaridze
in \cite{Cirq} as a new deductive tool for CoL and was further developed in \cite{Cirdeep,fromto,wenyan1,wenyan2} where a number of advantages of this novel sort of proof theory were revealed, such as high expressiveness, flexibility and efficiency.

In order to make this paper reasonably self-contained, in Section 2 we reproduce the basic concepts from \cite{Japfin,Japtam} on which the later parts of the paper will rely, including the old version of $\cst$ and its dual $\ccost$. An interested reader may consult \cite{Japfin,Japtam} for detailed explanations, illustrations and examples. In Section 3 we define the earlier-mentioned simplified version of $\cst$, and prove its equivalence to the old one. In Section 4 we prove that the set of principles validated by the operators $\{\neg,\wedge,\vee,\cst,\ccost\}$ is a proper superset of the set of those validated by $\{\neg,\wedge,\vee,\st,\cost\}$. Finally, in Section 5, we show that $\cst$ is strictly weaker than $\st$.

\section{Preliminaries}

 The letter $\wp$ is used as a variable ranging over $\{\top,\bot\}$, with $\neg\wp$ meaning $\wp$'s adversary. A {\bf move} is a finite string over standard keyboard alphabet. A {\bf labmove} is a move prefixed (``labeled'') with $\top$ or $\bot$. A {\bf run} is a finite or infinite sequence of labmoves, and a {\bf position} is a finite run. Runs are usually delimited by ``$\langle$" and ``$\rangle$", with $\langle\rangle$ thus denoting the {\bf empty run}. For any run $\Gamma$, $\neg\Gamma$ is the same as $\Gamma$, with the only difference that every label $\wp$ is changed to $\neg\wp$.

A {\bf game}\footnote{The concept of a game considered in CoL is more general than the one defined here, with games in the present sense called {\em constant games}. Since we (for simplicity) only consider constant games in this paper, we omit the word ``constant" and just say ``game".} is a pair $A=({\bf Lr}^{A},{\bf Wn}^{A})$, where: (1) ${\bf Lr}^{A}$ is a set of runs satisfying the condition that a finite or infinite run $\Gamma$ is in ${\bf Lr}^{A}$ iff so are all of $\Gamma$'s nonempty finite initial segments.\footnote{This condition can be seen to imply that
the empty run $\langle\rangle$ is always in ${\bf Lr}^{A}$.} If $\Gamma\in{\bf Lr}^{A}$, then $\Gamma$ is said to be a {\bf legal run} of $A$; otherwise $\Gamma$ is an {\bf illegal run} of $A$. A move $\alpha$ is a {\bf legal move} for a player $\wp$ in a position $\Phi$ of $A$ iff $\langle\Phi,\wp\alpha\rangle\in{\bf Lr}^{A}$; otherwise $\alpha$ is an {\bf illegal move}. When the last move of the shortest illegal initial segment of $\Gamma$ is $\wp$-labeled, $\Gamma$ is said to be a $\wp${\bf -illegal run} of $A$. (2) ${\bf Wn}^{A}$ is a function that sends every run $\Gamma$ to one of the players $\top$ or $\bot$, satisfying the condition that if $\Gamma$ is a $\wp$-illegal run of $A$, then ${\bf Wn}^{A}\langle\Gamma\rangle=\neg\wp$. When ${\bf Wn}^{A}\langle\Gamma\rangle=\wp$, $\Gamma$ is said to be a $\wp${\bf -won} run of $A$.\vspace{2mm}

The game operations dealt with in  the present paper are $\neg$ (negation), $\vee$ (parallel disjunction), $\wedge$ (parallel conjunction), $\cst$ (countable branching recurrence), $\ccost$ (countable branching corecurrence), $\st$ (uncountable branching recurrence) and $\cost$ (uncountable branching corecurrence).

Intuitively, $\neg$ is a role switch operator: $\neg A$ is the game $A$ with the roles of $\top$ and $\bot$ interchanged ($\top$'s legal moves and wins become those of $\bot$, and vice versa). Both $A\wedge B$ and $A\vee B$ are games playing which means playing the two components $A$ and $B$ simultaneously (in parallel). In $A\wedge B$, $\top$ is the winner if it wins in both components, while in $A\vee B$ winning in just one component is sufficient.

Next, as originally defined in \cite{Japfour}, a play of $\cst A$ (resp. $\ccost A$) starts as an ordinary play of $A$. At any time, however, the player $\bot$ (resp. $\top$) may make a ``replicative move" to create two copies of the current position $\Phi$ of $A$. This makes the game turn into two parallel games that continue from the same position $\Phi$.  The bits $0$ and $1$ are used to denote those two threads. Generally, at any time, $\bot$ (resp. $\top$) may (further) split any existing thread $w$ into two threads $w0$ and $w1$. Each thread in the eventual run of the game will be thus denoted by a (possibly infinite) bitstring, where a {\bf bitstring}\footnote{For bitstrings $x$ and $y$, we write $x\preceq y$ to mean that $x$ is a (not necessarily proper) initial segment (i.e. prefix) of $y$.} is a finite or infinite sequence of bits 0,1.  A bitstring $w$ is said to be {\bf essentially finite} if it contains only a finite number of ``1"s; otherwise $w$ is said to be {\bf essentially infinite}.  In $\cst A$, $\top$ is the winner if it wins $A$ in all infinite but essentially finite threads, while in $\ccost A$ winning in just one such thread is sufficient. Since there are only countably many essentially finite bitstrings, only countably many runs of $A$ are relevant when playing $\cst A$ or $\ccost A$. This is the intuitive explanation of what we called ``the old version" of $\cst$ (resp. $\ccost$) in the introduction.

Finally, the game $\st A$ (resp. $\cost A$) is the same as the game $\cst A$ (resp. $\ccost A$), with the only difference that, when determining the winner, all---essentially finite or essentially infinite---threads are relevant. Since there are uncountably many infinite bitstrings, uncountably many parallel runs of $A$ may be generated when playing $\st A$ or $\cost A$. We also call this version of $\st$ (resp. $\cost$) the ``old" version found in \cite{Jap03,Japfin}. Because recently a new simplified version of uncountable branching (co)recurrence was introduced in \cite{Japface}. It is different from yet equivalent to (in all relevant respects) the above old version. Specifically, both (the new versions of) $\st A$ and $\cost A$ are games playing which means simultaneously playing a continuum of copies (or ``threads") of $A$. Each copy/thread is denoted by an infinite bitstring and vice versa. Making a move $w.\alpha$, where $w$ is a finite bitstring, means making the move $\alpha$ simultaneously in all threads of the form $wy$. In $\st A$, $\top$ is the winner iff it wins in all threads of $A$, while in $\cost A$ winning in just one thread is sufficient. It should be noted that, when dealing with the uncountable branching (co)recurrence in this paper, we exclusively employ the new version of it.

Let $\Gamma$ be a run and $\alpha$ be a move. The notation $\Gamma^{\alpha}$
will be used to indicate the result of deleting from $\Gamma$ all moves (together with their labels) except those that look like $\alpha\beta$ for some move $\beta$, and then further deleting the prefix ``$\alpha$" from such moves. For instance, $\langle\top 1.\alpha,\ \bot 2.\beta,\ \top 1.\gamma,\ \bot 2.\delta\rangle^{1.}=\langle\top\alpha,\ \top\gamma\rangle$.

Let $\Theta$ be a run and $x$ be an infinite bitstring. The notation $\Theta^{\preceq x}$ will be used to indicate the result of deleting from $\Theta$ all moves (together with their labels) except those that look like $u.\beta$ for some move $\beta$ and some finite initial segment $u$ of $x$, and then further deleting the prefix ``u." from such moves. For instance, $\langle\bot 10.\alpha,\ \top 111.\beta,\ \bot 1.\gamma,\ \bot 00.\alpha\rangle^{\preceq 111\ldots}=\langle\top\beta,\ \bot\gamma\rangle$.


The earlier-outlined intuitive characterizations of the game operators are captured by the following formal definition.
Below, $A$, $A_1$, $A_2$ are arbitrary games, $\alpha$ ranges over moves, $i\in\{1,2\}$, $s$ ranges over finite bitstrings, $x$ ranges over infinite bitstrings, $\Gamma$
is an arbitrary run, and $\Omega$ is any legal run of the game that is being defined.\vspace{2mm}\\
1. $\neg A$ ({\bf negation}) is defined by:\vspace{2mm}

{\bf (i)} $\Gamma\in{\bf Lr}^{\neg A}$ iff $\neg\Gamma\in{\bf Lr}^{A}$.

{\bf (ii)} ${\bf Wn}^{\neg A}\langle\Omega\rangle=\top$ iff ${\bf Wn}^{A}\langle\neg\Omega\rangle=\bot$.\vspace{2mm}\\
2. $A_1\wedge A_2$ ({\bf parallel conjunction}) is defined by:\vspace{2mm}

{\bf (i)} $\Gamma\in{\bf Lr}^{A_1\wedge A_2}$ iff every move of $\Gamma$ is $i.\alpha$ for some $i$,$\alpha$ and, for both $i$, $\Gamma^{i.}\in{\bf Lr}^{A_i}$.

{\bf (ii)} ${\bf Wn}^{A_1\wedge A_2}\langle\Omega\rangle=\top$ iff, for both $i$, ${\bf Wn}^{A_i}\langle\Omega^{i.}\rangle=\top$.\vspace{2mm}\\
3. $A_1\vee A_2$ ({\bf parallel disjunction}) is defined by:\vspace{2mm}

{\bf (i)} $\Gamma\in{\bf Lr}^{A_1\vee A_2}$ iff every move of $\Gamma$ is $i.\alpha$ for some $i$,$\alpha$ and, for both $i$, $\Gamma^{i.}\in{\bf Lr}^{A_i}$.

{\bf (ii)} ${\bf Wn}^{A_1\vee A_2}\langle\Omega\rangle=\top$ iff, for some $i$, ${\bf Wn}^{A_i}\langle\Omega^{i.}\rangle=\top$.\vspace{2mm}\\
4. $\st A$ ({\bf uncountable branching recurrence})\footnote{This is the formal definition of the new simplified version of $\st$ introduced in \cite{Japface}. The same applies to the following definition of $\cost$.} is defined by: \vspace{2mm}

{\bf (i)} $\Gamma\in{\bf Lr}^{\sti A}$ iff every move of $\Gamma$ is $s.\alpha$ for some $s$,$\alpha$ and, for all $x$, $\Gamma^{\preceq x}\in{\bf Lr}^{A}$.

{\bf (ii)} ${\bf Wn}^{\sti A}\langle\Omega\rangle=\top$ iff, for all $x$, ${\bf Wn}^{A}\langle\Omega^{\preceq x}\rangle=\top$.\vspace{2mm}\\
5. $\cost A$ ({\bf uncountable branching corecurrence}) is defined by: \vspace{2mm}

{\bf (i)} $\Gamma\in{\bf Lr}^{\costi A}$ iff every move of $\Gamma$ is $s.\alpha$ for some $s$,$\alpha$ and, for all $x$, $\Gamma^{\preceq x}\in{\bf Lr}^{A}$.

{\bf (ii)} ${\bf Wn}^{\costi A}\langle\Omega\rangle=\top$ iff, for some $x$, ${\bf Wn}^{A}\langle\Omega^{\preceq x}\rangle=\top$.\vspace{2mm}\\
6. $\cst A$ ({\bf countable branching recurrence})\footnote{This is the formal definition of the ``old" version of $\cst$.  The same applies to the following definition of $\ccost$.} is defined as follows. There are two types of legal moves in legal positions of $\cst A$: {\bf replicative} and {\bf non-replicative}. What is called a {\bf node of the underlying BT-structure} of $\langle\Phi\rangle\cst A$, where $\Phi$ is a position, is a bitstring $w$ such that $w$ is either empty, or else is $u0$ or $u1$ for some bitstring $u$ such that $\Phi$ contains the move $u$:. Such a node is said to be a {\bf leaf} iff it is not a proper prefix of any other node of the underlying BT-structure of $\langle\Phi\rangle\cst A$. There are two sorts of legal moves in every position: replicative and non-replicative. A replicative move can only be made by $\bot$, and such a move in a given position $\Phi$ should be $w$:, where $w$ is a leaf of the underlying BT-structure of $\langle\Phi\rangle\cst A$. As for non-replicative moves, they can be made by either player. Such a move by a player $\wp$ in a given position $\Phi$ should be $w.\alpha$, where $w$ is a node of the underlying BT-structure of $\langle\Phi\rangle\cst A$ and $\alpha$ is a move such that, for any infinite extension $v$ of $w$, $\alpha$ is a legal move by $\wp$ in the position $\Phi^{\preceq v}$ of $A$.  A legal run $\Gamma$ of $\cst A$ is won by $\top$ iff, for every infinite but essentially finite bitstring $v$, $\Gamma^{\preceq v}$ is a $\top$-won run of $A$.\vspace{2mm}\\
7. $\ccost A$ ({\bf countable branching corecurrence}) is defined in a symmetric way to $\cst A$, by interchanging $\top$ with $\bot$. Equivalently, it can be simply defined by $\ccost A=\neg\cst\neg A$.\vspace{3mm}

In what follows,  we explain---formally or informally---several additional concepts relevant to our proofs.

(1) {\bf Static games}: CoL restricts its attention to a special yet very wide subclass of games termed ``static". Intuitively, static games are interactive tasks where the relative speeds of the players are irrelevant, as it never hurts a player to postpone making moves. For either player $\wp$, a run $\Omega$ is said to be a {\bf $\wp$-delay} of a run $\Gamma$ iff for both players $\wp'\in\{\top,\bot\}$, the subsequence of $\wp'$-labeled moves of $\Omega$ is the same as that of $\Gamma$, and for any $n,k\geq 1$, if the $n$'th $\wp$-labeled move is made later than (is to the right of) the $k$'th $\neg\wp$-labeled move in $\Gamma$, then so is it in $\Omega$. For instance, the run $\langle\bot\alpha,\top\alpha,\top\gamma,\bot\beta\rangle$ is a $\bot$-delay of the run $\langle\bot\alpha,\top\alpha,\bot\beta,\top\gamma\rangle$. A run is said to be {\bf $\wp$-legal} iff it is not $\wp$-illegal. Finally, a game $A$ is said to be {\bf static} iff, whenever a run $\Omega$ is a $\wp$-delay of a run $\Gamma$, we have: if $\Gamma$ is a $\wp$-legal run of $A$, then so is $\Omega$; if $\Gamma$ is a $\wp$-won run of $A$, then so is $\Omega$. It is known (\cite{Jap03,Japfin,Japface}) that the class of static games is closed under the operations $\neg,\wedge,\vee,\cst,\ccost,\st,\cost$ (as well as any other game operations studied in CoL).

(2) {\bf EPM} and {\bf BMEPM}: CoL understands $\top$'s effective strategies as interactive machines.
Several sorts of such machines have been proposed and studied in CoL, all of them turning out to be equivalent in computing power once we exclusively consider static games. In this paper we will use two sorts of such machines, called the {\em easy-play machine} ({\bf EPM}) and the {\em block-move EPM} ({\bf BMEPM}). Both of them are sorts of Turing machines with the additional capability of making moves, and have two tapes\footnote{Often there is also a third tape called the {\em valuation tape}. Its function is to provide values for the variables on which a game may depend. However, as we remember, in this paper we only consider constant games --- games that do not depend on any variables. This makes it possible to safely remove the valuation tape (or leave it there but fully ignore), as this tape is no longer relevant.}: the ordinary read/write {\em work tape}, and the read-only {\em run tape}.  The run tape serves as a dynamic input, at any time (``{\bf clock cycle}") spelling the current position: every time one of the players makes a move, that move---with the corresponding label---is automatically appended to the content of this tape.  An EPM is the machine where either player can make at most one move on a given clock cycle, but the environment can move only when the machine explicitly allows it to do so (this sort of an action is called {\bf granting permission} ); an BMEPM only differs from an EPM in that either player can make any finite number of moves at once.\footnote{In another more basic sort of machines
 called the {\em hard-play machines} ({\bf HPM}), the machine can make at most one move at any time but the environment can make any number of moves (needing no ``permission'' for that). }

(3) {\bf Strategies}:
Let ${\cal M}$ be an EPM or BMEPM. A {\em configuration} of ${\cal M}$ is a full description of the current state of the machine, the contents of its two tapes, and the locations of the corresponding two scanning heads. The {\em initial configuration } is the configuration where ${\cal M}$ is in its start state and both tapes are empty. A configuration $C'$ is said to be an {\em successor} of a configuration $C$ if $C'$ can legally follow $C$ in the standard sense, based on the (deterministic) transition function of the machine and accounting for the possibility of nondeterministic updates of the content of the run tape. A {\bf computation branch} of ${\cal M}$ is a sequence of configurations of ${\cal M}$ where the first configuration is the initial configuration, and each other configuration is a successor of the previous one. Each computation branch $B$ of ${\cal M}$ incrementally spells a run $\Gamma$ on the run tape, which is called the {\bf run spelled by} $B$. Subsequently, any such run $\Gamma$ will be referred to as a {\bf run generated by} ${\cal M}$. A computation branch $B$ of ${\cal M}$ is said to be {\bf fair} iff, in it, permission has been granted infinitely many times. An {\bf algorithmic solution} ({\bf $\top$'s winning strategy}) for a given game $A$ is understood as an EPM or BMEPM ${\cal M}$ such that, whenever $B$ is a computation branch of ${\cal M}$ and $\Gamma$ the run spelled by $B$, $\Gamma$ is a $\top$-won run of $A$, where $B$ should be fair unless $\Gamma$ is a $\bot$-illegal run of $A$. When the above is the case, we say that ${\cal M}$ {\bf wins} $A$.
It is known (\cite{Japtowards}) that the two sorts of machines win the same static games. And since all games we ever deal with in this paper are static, in the following we may simply say ``a machine ${\cal M}$" without being specific about whether it is an EPM or BMEPM.\vspace{2mm}

Now about formulas and the underlying semantics. We have some fixed set of syntactic objects, called {\bf atoms}, for which $P$, $Q$, $R$ will be used as metavariables.
A {\bf formula} is built from atoms in the standard way
using the connectives $\neg$,$\vee$,$\wedge$,$\cst$,$\ccost$,$\st$,$\cost$, with $F\rightarrow G$ understood as an
abbreviation for $\neg F\vee G$ and $\neg$ limited only to atoms, where $\neg\neg F$ is understood as $F$, $\neg(F\wedge G)$ as $\neg F\vee\neg G$, $\neg(F\vee G)$ as $\neg F\wedge\neg G$, $\neg\cst F$ as $\ccost\neg F$, $\neg\ccost F$ as $\cst\neg F$, $\neg\st F$ as $\cost\neg F$, and $\neg\cost F$ as $\st\neg F$. A {\bf $(\neg,\wedge,\vee,\cst,\ccost)$-formula} is one not containing $\st$,$\cost$. Similarly, a {\bf $(\neg,\wedge,\vee,\st,\cost)$-formula} is one not containing $\cst$,$\ccost$. An {\bf interpretation} is a function $^*$ that sends every atom $P$ to a static game $P^*$, and extends to all formulas by seeing the logical connectives as the same-name game operations.
A formula $F$ is {\bf uniformly valid}, symbolically $\uvalid F$, iff there is a machine ${\cal M}$, called a {\bf uniform solution} of $F$, such that, for every interpretation $^*$, ${\cal M}$ wins $F^*$.\footnote{Another sort of validity studied in CoL is multiform validity. A formula $F$ is {\bf multiformly valid} iff, for every interpretation $^*$, there is a machine that wins $F^*$. Since uniform validity is stronger than multiform validity, all soundness-style results that we are going to establish about uniform validity automatically extend to multiform validity as well. Partly for this reason, in this paper we will be exclusively interested in uniform validity.}\vspace{2mm}


As noted in Section 1, {\bf CL15} is built in {\em cirquent calculus} for the basic $(\neg,\wedge,\vee,\st,\cost)$-fragment of CoL, whose formalism goes beyond formulas. In what follows in this paragraph, by a ``formula", we mean one of the $(\neg,\wedge,\vee,\st,\cost)$-formulas.  A {\bf cirquent} is a triple $C=(\vec{F},\vec{U},\vec{O})$ where: (1) $\vec{F}$ is a nonempty finite sequence of formulas, whose elements are said to be the {\bf oformulas} of $C$. Here the prefix ``o" is used to mean a formula together with a particular occurrence of it in $\vec{F}$. For instance, if $\vec{F}=\langle G, H, H\rangle$, then the cirquent has three oformulas while only two formulas. (2) Both $\vec{U}$ and $\vec{O}$ are nonempty finite sequences of nonempty sets of oformulas of $C$. The elements of $\vec{U}$ are said to be the {\bf undergroups} of $C$, and the elements of $\vec{O}$ are said to be the {\bf overgroups} of $C$. Again, two undergroups (resp. overgroups) may be identical as sets (have identical {\bf contents}), yet they count as different undergroups (resp. overgroups) because they occur at different places in $\vec{U}$ (resp. $\vec{O}$). (3) Additionally, every oformula is required to be in at least one undergroup and at least one overgroup.

Rather than writing cirquents as ordered tuples in the above style, we prefer to represent them through (and identify them with) {\bf diagrams}. Below is such a representation for the cirquent that has four oformulas $H, F, E, F$, three undergroups $\{H,F\}$, $\{F,E\}$, $\{F\}$ and three overgroups $\{H,F,E\}$, $\{E\}$, $\{F\}$.

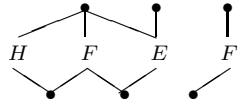
\begin{center}
\begin{picture}(80,50)(0,17)\footnotesize
\put(-10,38){$H\ \ \ \ \ \ \ F\ \ \ \ \ \ \ E\ \ \ \ \ \ \ F$}
\put(-8,35){\line(3,-2){14}}\put(20,35){\line(-3,-2){14}}\put(4,23){$\bullet$}
\put(20,35){\line(3,-2){14}}\put(46,35){\line(-3,-2){14}}\put(32,23){$\bullet$}
\put(74,35){\line(-3,-2){14}}\put(58,23){$\bullet$}
\put(19,57){\line(-5,-2){25}}\put(19,57){\line(0,-1){10}}
\put(19,57){\line(5,-2){27}}\put(17,56){$\bullet$}
\put(44,56){$\bullet$}\put(71,56){$\bullet$}
\put(46,57){\line(0,-1){10}}\put(73,57){\line(0,-1){10}}
\end{picture}
\end{center}
Each group in the cirquent/diagram is represented by (and identified with) a $\bullet$, where the {\bf arcs} (lines connecting the $\bullet$ with oformulas) are pointing to the oformulas that the given group contains.\vspace{2mm}

There are ten inference rules in {\bf CL15}. Below we reproduce those rules from \cite{Japtam} with $\st$ and $\cost$ rewritten as $\cst$ and $\ccost$, respectively. To semantically differentiate the two versions of {\bf CL15} (when necessary), we may use the name ${\bf CL15}(\st)$ for the system that understands (and writes) the recurrence operator as $\st$, and use ${\bf CL15}(\cst)$ for the system that understands (and writes) the recurrence operator as $\cst$. Correspondingly, throughout the rest of this section, by a ``formula", we mean one of the $(\neg,\wedge,\vee,\cst,\ccost)$-formulas.\vspace{2mm}

{\bf Axiom (A):} Axiom is a ``rule" with no premises. It introduces the cirquent

\ \ \ \ \ \ \ $(\langle\neg F_1,F_1,\ldots,\neg F_n,F_n\rangle, \langle\{\neg F_1,F_1\},\ldots,\{\neg F_n,F_n\}\rangle, \langle\{\neg F_1,F_1\},\ldots,\{\neg F_n,F_n\}\rangle)$,\\
where $n$ is any positive integer, and $F_1,\ldots,F_n$ are any formulas. All rules other than Axiom take a single premise.

{\bf Exchange (E):} This rule comes in three versions: {\bf Undergroup Exchange}, {\bf Oformula Exchange} and {\bf Overgroup Exchange}.
The conclusion of Oformula Exchange is obtained by interchanging in the premise two adjacent oformulas $E$ and $F$, and
redirecting to $E$ (resp. $F$) all arcs that were originally pointing to $E$ (resp. $F$). Undergroup (resp. Overgroup) Exchange is the same, with the only difference that the objects interchanged are undergroups (resp. overgroups).

{\bf Duplication (D):} This rule comes in two versions: {\bf Undergroup Duplication} and {\bf Overgroup Duplication}. The conclusion of Undergroup Duplication is obtained by replacing in the premise some undergroup $U$ with two adjacent undergroups whose contents are identical to that of $U$. Similarly for Overgroup Duplication.

{\bf Merging (M):} The conclusion of this rule can be obtained from the premise by merging any two adjacent overgroups $O_1$ and $O_2$ into one overgroup $O$, and including in $O$ all oformulas that were originally contained in $O_1$ or $O_2$ or both.

{\bf Weakening (W):} For the convenience of description, we explain this and the remaining rules in the bottom-up view. The premise of this rule is obtained by deleting in the conclusion an arc between some undergroup $U$ with $\geq 2$ elements and some oformula $F$; if $U$ was the only undergroup containing $F$, then $F$ should also be deleted, together with all arcs between $F$ and overgroups; if such a deletion makes some overgroups empty, then they should also be deleted.

{\bf Contraction (C):} The premise of this rule is obtained by replacing in the conclusion an oformula $\ccost F$ by two adjacent oformulas $\ccost F$ and $\ccost F$, and including both of them in exactly the same undergroups and overgroups in which the original $\ccost F$ was contained.


{\bf Disjunction introduction ($\vee$):} The premise of this rule is obtained by replacing in the conclusion an oformula $E\vee F$ by two adjacent oformulas $E$ and $F$, and including both of them in exactly the same undergroups and overgroups in which the original $E\vee F$ was contained.

{\bf Conjunction introduction ($\wedge$):} According to this rule, if a cirquent (the conclusion) has an oformula $E\wedge F$, then the premise
can be obtained by splitting the original $E\wedge F$ into two adjacent oformulas $E$ and $F$, including both of them in exactly the same overgroups in which the original $E\wedge F$ was contained, and splitting every undergroup $\Gamma$ that originally contained $E\wedge F$ into two adjacent undergroups $\Gamma^{E}$ and $\Gamma^{F}$, where $\Gamma^{E}$ contains $E$ (but not $F$), and $\Gamma^{F}$ contains $F$ (but not $E$), with all other ($\neq E\wedge F$) oformulas of $\Gamma$ contained by both $\Gamma^{E}$ and $\Gamma^{F}$.

{\bf Recurrence introduction ($\cst$):} The premise of this rule is obtained by replacing in the conclusion an oformula $\cst F$ by $F$, with all arcs unchanged, and inserting a new overgroup $\Gamma$ that contains $F$ as its {\em only} oformula.

{\bf Corecurrence introduction ($\ccost$):} The premise of this rule is obtained by replacing in the conclusion an oformula $\ccost F$ by $F$, with all arcs unchanged, and additionally including $F$ in any (possibly zero) number of the already existing overgroups.\vspace{2mm}

Below we provide illustrations for all rules, in each case an abbreviated name of the rule standing next to the horizontal line separating the premise from the conclusion. Our illustration for the axiom (the ``{\bf A}" labeled rule) is a specific cirquent where $n=2$; our illustrations for all other rules are merely examples chosen arbitrarily. Unfortunately, no systematic ways for schematically representing cirquent calculus rules have been elaborated so far. This explains why we appeal to examples instead.

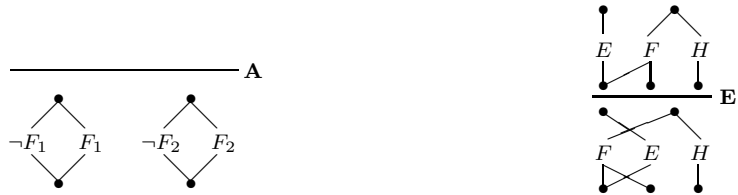
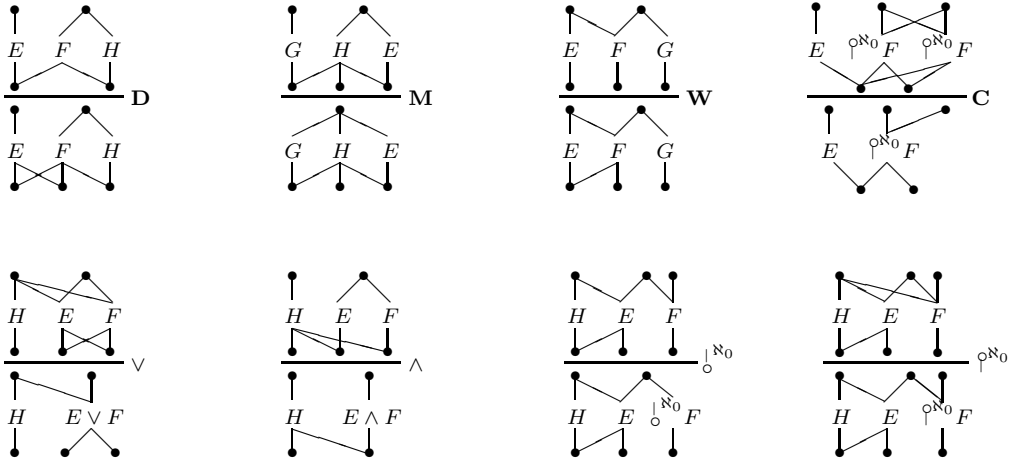
\begin{center}
\begin{picture}(80,90)(60,7)\footnotesize
\put(-70,0){\begin{picture}(80,80)
\put(20,60){\line(1,0){86}}\put(19,30){$\neg F_1\ \ \ \ F_1\ \ \ \ \ \neg F_2\ \ \ \ F_2$}\put(108,57){\bf A}
\put(28,38){\line(1,1){10}}\put(48,38){\line(-1,1){10}}
\put(28,28){\line(1,-1){10}}\put(48,28){\line(-1,-1){10}}
\put(36,15){$\bullet$}\put(36,47){$\bullet$}
\put(78,38){\line(1,1){10}}\put(98,38){\line(-1,1){10}}
\put(78,28){\line(1,-1){10}}\put(98,28){\line(-1,-1){10}}
\put(86,15){$\bullet$}\put(86,47){$\bullet$}
\end{picture}}
\put(40,0){\begin{picture}(80,80)\put(130,50){\line(1,0){45}}\put(178,47){\bf E}
\put(131,65){$E\ \ \ \ F\ \ \ \ H$}\put(134,63){\line(0,-1){8}}\put(152,63){\line(0,-1){8}}\put(170,63){\line(0,-1){8}}
\put(132,52){$\bullet$}\put(150,52){$\bullet$}\put(168,52){$\bullet$}\put(152,63){\line(-2,-1){18}}
\put(132,81){$\bullet$}\put(134,74){\line(0,1){8}}\put(159,81){$\bullet$}\put(161,83){\line(-1,-1){10}}\put(161,83){\line(1,-1){10}}

\put(131,26){$F\ \ \ \ E\ \ \ \ H$}
\put(132,42){$\bullet$}\put(159,42){$\bullet$}
\put(134,44){\line(3,-2){16}}\put(161,44){\line(1,-1){10}}\put(161,44){\line(-5,-2){25}}
\put(132,13){$\bullet$}\put(150,13){$\bullet$}\put(168,13){$\bullet$}
\put(152,24){\line(-2,-1){18}}\put(134,24){\line(0,-1){8}}\put(134,24){\line(2,-1){18}}\put(170,24){\line(0,-1){8}}
\end{picture}}

\end{picture}
\end{center}

\begin{center}
\begin{picture}(80,80)(60,20)\footnotesize

\put(-220,5){\begin{picture}(80,80)\put(130,50){\line(1,0){45}}\put(178,47){\bf D}
\put(0,0){\begin{picture}(80,80)\put(131,65){$E\ \ \ \ F\ \ \ \ H$}\put(134,63){\line(0,-1){8}}\put(170,63){\line(0,-1){8}}
\put(132,52){$\bullet$}\put(168,52){$\bullet$}\put(152,63){\line(-2,-1){18}}\put(152,63){\line(2,-1){18}}
\put(132,81){$\bullet$}\put(134,74){\line(0,1){8}}\put(159,81){$\bullet$}\put(161,83){\line(-1,-1){10}}\put(161,83){\line(1,-1){10}}\end{picture}}

\put(0,-38){\begin{picture}(80,80)\put(131,65){$E\ \ \ \ F\ \ \ \ H$}\put(134,63){\line(0,-1){8}}\put(170,63){\line(0,-1){8}}
\put(132,52){$\bullet$}\put(168,52){$\bullet$}\put(152,63){\line(-2,-1){18}}\put(152,63){\line(0,-1){8}}\put(150,52){$\bullet$}
\put(132,81){$\bullet$}\put(134,74){\line(0,1){8}}\put(159,81){$\bullet$}\put(161,83){\line(-1,-1){10}}\put(161,83){\line(1,-1){10}}
\put(134,63){\line(2,-1){18}}\put(152,63){\line(2,-1){18}}
\end{picture}}
\end{picture}}

\put(-115,5){\begin{picture}(80,80)\put(130,50){\line(1,0){45}}\put(178,47){\bf M}\put(152,63){\line(-2,-1){18}}\put(152,63){\line(2,-1){18}}
\put(0,0){\begin{picture}(80,80)\put(131,65){$G\ \ \ \ H\ \ \ \ E$}\put(134,63){\line(0,-1){8}}\put(170,63){\line(0,-1){8}}
\put(132,52){$\bullet$}\put(168,52){$\bullet$}\put(152,63){\line(0,-1){8}}\put(150,52){$\bullet$}
\put(132,81){$\bullet$}\put(134,74){\line(0,1){8}}\put(159,81){$\bullet$}\put(161,83){\line(-1,-1){10}}\put(161,83){\line(1,-1){10}}\end{picture}}

\put(0,-38){\begin{picture}(80,80)\put(131,65){$G\ \ \ \ H\ \ \ \ E$}\put(134,63){\line(0,-1){8}}\put(170,63){\line(0,-1){8}}\put(152,63){\line(-2,-1){18}}\put(152,63){\line(2,-1){18}}
\put(132,52){$\bullet$}\put(168,52){$\bullet$}\put(152,63){\line(0,-1){8}}\put(150,52){$\bullet$}
\put(150,81){$\bullet$}\put(152,82){\line(0,-1){8}}\put(152,82){\line(-2,-1){18}}\put(152,82){\line(2,-1){18}}\end{picture}}
\end{picture}}

\put(-10,5){\begin{picture}(80,80)\put(130,50){\line(1,0){45}}\put(178,47){\bf W}
\put(0,0){\begin{picture}(80,80)\put(131,65){$E\ \ \ \ F\ \ \ \ G$}\put(134,63){\line(0,-1){8}}\put(134,82){\line(2,-1){17}}\put(170,63){\line(0,-1){8}}
\put(132,52){$\bullet$}\put(168,52){$\bullet$}\put(152,63){\line(0,-1){8}}\put(150,52){$\bullet$}
\put(132,81){$\bullet$}\put(134,74){\line(0,1){8}}\put(159,81){$\bullet$}\put(161,83){\line(-1,-1){10}}\put(161,83){\line(1,-1){10}}\end{picture}}

\put(0,-38){\begin{picture}(80,80)\put(131,65){$E\ \ \ \ F\ \ \ \ G$}\put(134,63){\line(0,-1){8}}\put(152,63){\line(-2,-1){18}}\put(134,82){\line(2,-1){17}}\put(170,63){\line(0,-1){8}}
\put(132,52){$\bullet$}\put(168,52){$\bullet$}\put(152,63){\line(0,-1){8}}\put(150,52){$\bullet$}
\put(132,81){$\bullet$}\put(134,74){\line(0,1){8}}\put(159,81){$\bullet$}\put(161,83){\line(-1,-1){10}}\put(161,83){\line(1,-1){10}}\end{picture}}
\end{picture}}

\put(90,5){\begin{picture}(80,80)\put(124,50){\line(1,0){60}}\put(186,47){\bf C}
\put(0,0){\begin{picture}(80,80)\put(124,65){$E\ \  \ \ccost F\ \ \ \ccost F $}\put(129,63){\line(5,-3){15}}\put(153,63){\line(-1,-1){10}}\put(153,63){\line(1,-1){10}}\put(178,63){\line(-5,-3){16}}\put(178,63){\line(-4,-1){34}}

\put(142,51){$\bullet$}\put(160,51){$\bullet$}\put(176,83){\line(-5,-2){24}}
\put(125,82){$\bullet$}\put(127,84){\line(0,-1){10}}\put(150,82){$\bullet$}\put(152,84){\line(0,-1){10}}
\put(153,83){\line(5,-2){23}}
\put(174,82){$\bullet$}\put(176,84){\line(0,-1){10.5}}\end{picture}}

\put(0,-38){\begin{picture}(80,80)\put(129,65){$E\ \ \ \  \ccost F$}\put(134,63){\line(1,-1){10}}\put(154,63){\line(-1,-1){10}}\put(154,63){\line(1,-1){10}}
\put(142,51){$\bullet$}\put(162,51){$\bullet$}
\put(130,81){$\bullet$}\put(132,74){\line(0,1){8}}\put(152,81){$\bullet$}\put(154,83){\line(0,-1){9}}
\put(174,81){$\bullet$}\put(176,83){\line(-5,-2){22}}
\end{picture}}
\end{picture}}
\end{picture}
\end{center}

\begin{center}
\begin{picture}(80,90)(60,20)\footnotesize

\put(-220,5){\begin{picture}(80,80)\put(130,50){\line(1,0){45}}\put(178,47){\bf $\vee$}
\put(0,0){\begin{picture}(80,80)\put(131,65){$H\ \ \ \ E\ \ \ \ F$}\put(134,63){\line(0,-1){8}}\put(170,63){\line(0,-1){8}}\put(170,63){\line(-2,-1){18}}\put(152,63){\line(2,-1){18}}
\put(132,52){$\bullet$}\put(168,52){$\bullet$}\put(152,63){\line(0,-1){8}}\put(150,52){$\bullet$}
\put(132,81){$\bullet$}\put(134,74){\line(0,1){8}}\put(159,81){$\bullet$}\put(161,83){\line(-1,-1){10}}\put(161,83){\line(1,-1){10}}
\put(133,82){\line(2,-1){18}}\put(133,82){\line(4,-1){38}}\end{picture}}

\put(0,-38){\begin{picture}(80,80)\put(131,65){$H\ \ \ \ \ E\vee F$}\put(134,63){\line(0,-1){8}}\put(163,63){\line(1,-1){10}}

\put(132,52){$\bullet$}\put(170,52){$\bullet$}\put(151,52){$\bullet$}
\put(132,81){$\bullet$}\put(134,74){\line(0,1){8}}\put(161,81){$\bullet$}\put(163,82){\line(0,-1){9}}
\put(163,63){\line(-1,-1){10}}
\put(133,83){\line(3,-1){30}}
\end{picture}}
\end{picture}}

\put(-115,5){\begin{picture}(80,80)\put(130,50){\line(1,0){45}}\put(178,47){\bf $\wedge$}
\put(0,0){\begin{picture}(80,80)\put(131,65){$H\ \ \ \ E\ \ \ \ F$}\put(134,63){\line(0,-1){8}}
\put(134,63){\line(2,-1){18}}\put(134,63){\line(4,-1){36}}\put(170,63){\line(0,-1){8}}
\put(132,52){$\bullet$}\put(168,52){$\bullet$}\put(152,63){\line(0,-1){8}}\put(150,52){$\bullet$}
\put(132,81){$\bullet$}\put(134,74){\line(0,1){8}}\put(159,81){$\bullet$}\put(161,83){\line(-1,-1){10}}\put(161,83){\line(1,-1){10}}\end{picture}}

\put(0,-38){\begin{picture}(80,80)\put(131,65){$H\ \ \ \ \ E\wedge F$}\put(134,63){\line(0,-1){8}}\put(163,63){\line(0,-1){8}}\put(134,63){\line(3,-1){28}}
\put(132,52){$\bullet$}\put(161,52){$\bullet$}
\put(132,81){$\bullet$}\put(134,74){\line(0,1){8}}\put(161,81){$\bullet$}\put(163,83){\line(0,-1){9}}\end{picture}}
\end{picture}}

\put(-8,5){\begin{picture}(80,80)\put(130,50){\line(1,0){50}}\put(182,47){\bf $\cst$}
\put(0,0){\begin{picture}(80,80)\put(131,65){$H\ \ \ \ E\ \ \ \ F$}\put(134,63){\line(0,-1){8}}\put(152,63){\line(-2,-1){18}}\put(134,82){\line(2,-1){17}}\put(171,63){\line(0,-1){8}}
\put(132,52){$\bullet$}\put(169,52){$\bullet$}\put(152,63){\line(0,-1){8}}\put(150,52){$\bullet$}\put(169,81){$\bullet$}\put(171,82){\line(0,-1){9}}
\put(132,81){$\bullet$}\put(134,74){\line(0,1){8}}\put(159,81){$\bullet$}\put(161,83){\line(-1,-1){10}}\put(161,83){\line(1,-1){10}}\end{picture}}

\put(0,-38){\begin{picture}(80,80)\put(131,65){$H\ \ \ \ E\ \ \cst F$}\put(134,63){\line(0,-1){8}}\put(152,63){\line(-2,-1){18}}\put(134,82){\line(2,-1){17}}\put(171,63){\line(0,-1){8}}
\put(132,52){$\bullet$}\put(169,52){$\bullet$}\put(152,63){\line(0,-1){8}}\put(150,52){$\bullet$}
\put(132,81){$\bullet$}\put(134,74){\line(0,1){8}}\put(159,81){$\bullet$}\put(161,83){\line(-1,-1){10}}\put(161,83){\line(5,-4){10}}\end{picture}}
\end{picture}}

\put(95,5){\begin{picture}(80,80)\put(125,50){\line(1,0){55}}\put(183,47){\bf $\ccost$}
\put(-3,0){\begin{picture}(80,80)\put(131,65){$H\ \ \ \ E\ \ \ \ F$}\put(134,63){\line(0,-1){8}}\put(152,63){\line(-2,-1){18}}\put(134,82){\line(2,-1){17}}\put(171,63){\line(0,-1){8}}\put(134,82){\line(4,-1){37}}
\put(132,52){$\bullet$}\put(169,52){$\bullet$}\put(152,63){\line(0,-1){8}}\put(150,52){$\bullet$}\put(169,81){$\bullet$}\put(171,82){\line(0,-1){9}}
\put(132,81){$\bullet$}\put(134,74){\line(0,1){8}}\put(159,81){$\bullet$}\put(161,83){\line(-1,-1){10}}\put(161,83){\line(1,-1){10}}\end{picture}}

\put(-3,-38){\begin{picture}(80,80)\put(131,65){$H\ \ \ \ E\ \ \ \ccost F$}\put(134,63){\line(0,-1){8}}\put(152,63){\line(-2,-1){18}}\put(134,82){\line(2,-1){17}}
\put(173,63){\line(0,-1){8}}\put(171,81){$\bullet$}\put(173,82){\line(0,-1){9}}
\put(132,52){$\bullet$}\put(171,52){$\bullet$}\put(152,63){\line(0,-1){8}}\put(150,52){$\bullet$}
\put(132,81){$\bullet$}\put(134,74){\line(0,1){8}}\put(159,81){$\bullet$}\put(161,83){\line(-1,-1){10}}\put(161,83){\line(5,-4){12}}\end{picture}}
\end{picture}}
\end{picture}
\end{center}

\vspace{2mm}
The above are all ten rules of {\bf CL15$(\cst)$}. A {\bf CL15$(\cst)$-proof} (or simply a {\bf proof}) of a cirquent $C$ is a sequence $\langle C_1,\ldots,C_n\rangle$ of cirquents, where $n\geq 1$, such that $C_n=C$, $C_1$ is an axiom, and $C_{i}$ ($1< i\leq n$) follows from $C_{i-1}$ by one of the rules of {\bf CL15$(\cst)$}. For any formula $F$, the expression $F^{\clubsuit}$ is used to denote the cirquent $(\langle F\rangle,\langle\{F\}\rangle,\langle\{F\}\rangle)$. Then a {\bf CL15$(\cst)$-proof} (or simply a {\bf proof}) of a formula $F$ is stipulated to be  a proof of the cirquent $F^{\clubsuit}$. A formula or cirquent $X$ is {\bf provable}, symbolically {\bf CL15}$(\cst)\vdash X$, iff it has a proof.

As mentioned, {\bf CL15}$(\st)$ is the same as {\bf CL15}$(\cst)$, only with $\st,\cost$ instead of $\cst,\ccost$.

\begin{theorem}\label{niu}
{\em (Japaridze \cite{Japtam,Japtam2})} A $(\neg,\wedge,\vee,\st,\cost)$-formula is uniformly valid iff it is provable in {\bf CL15}$(\st)$.
\end{theorem}

\section{A new version of the countable branching recurrence}

As we have seen in the preceding section, the existing definition of $\cst$ is relatively intricate, which considerably impedes the task of understanding this sort of recurrence. So, in this section we introduce a new simplified---yet equivalent to the old---version of the countable branching recurrence. In order to avoid confusion, when necessary, we shall use $\cst_L$, $\ccost_L$ for the new versions of $\cst$,$\ccost$, and use $\cst_T$, $\ccost_T$ for the old ones. The same notation applies to any formula $F$, where $F_T$ is the result of replacing in $F$ all occurrences of $\cst$ (resp. $\ccost$) by $\cst_T$ (resp. $\ccost_T$), and $F_L$ is the result of replacing in $F$ all occurrences of $\cst$ (resp. $\ccost$) by $\cst_L$ (resp. $\ccost_L$). Here, as understood, we extend the earlier-defined concept of a formula so that now a formula may contain either version $\cst_T$,$\cst_L$ of $\cst$ and/or either version $\ccost_T$,$\ccost_L$ of $\ccost$. The semantics of formulas and, particularly, the concept of uniform validity extend to this broader class of formulas in a straightforward/expected way.

As mentioned earlier, the old version of $\cst$ only differs from the old version of $\st$ in that, when determining the winner, only essentially finite threads are relevant. On the other hand, the paper \cite{Japface} has completed the task of replacing the old ``canonical" definition of $\st$ by a new, simple and compact, definition of $\st$ as we have seen in Section 2. For these reasons, the new definition of $\cst$, that we will introduce in the following, follows the same idea of the new definition of $\st$, only with ``infinite but essentially finite bitstrings" instead of ``infinite bitstrings", when determining the winner.

\begin{definition}\label{def} Below $A$ is an arbitrary game, $\alpha$ ranges over moves,  $w$ ranges over finite bitstrings, $x$ ranges over infinite bitstrings, $v$ ranges over infinite but essentially finite bitstrings, $\Gamma$ is any run, and $\Omega$ is any legal run of the game that is being defined.\vspace{9pt}

\noindent 1. $\cst_L A$ is defined by:
\begin{quote}\begin{description}
\item[(i)] $\Gamma\in \legal{\sti_L^{\aleph_0} A}{}$ iff every move of $\Gamma$ is $w.\alpha$ for some $w,\alpha$ and, for all $x$, $\Gamma^{\preceq x}\in\legal{A}{}$.
\item[(ii)] $\win{\sti_L^{\aleph_0} A}{}\seq{\Omega}= \pp$ iff, for all $v$,  $\win{A}{}\seq{\Omega^{\preceq v}}= \pp$. \vspace{5pt}
\end{description}\end{quote}

\noindent 2. $\ccost_L A$ is defined by:
\begin{quote}\begin{description}
\item[(i)] $\Gamma\in \legal{\costi_L^{\aleph_0} A}{}$ iff every move of $\Gamma$ is $w.\alpha$ for some $w,\alpha$ and, for all $x$, $\Gamma^{\preceq x}\in\legal{A}{}$.
\item[(ii)] $\win{\costi_L^{\aleph_0} A}{}\seq{\Omega}= \pp$ iff, for some $v$,  $\win{A}{}\seq{\Omega^{\preceq v}}= \pp$. \vspace{5pt}
\end{description}\end{quote}

\end{definition}

It is obvious that $\ccost_L$ is the dual operation of $\cst_L$ with $\ccost_L A=\neg\cst_L\neg A$. In what follows, we first prove that $\cst_L$ and $\ccost_L$ preserve the static property of games, and then show that $\cst_L$ and $\ccost_L$ are logically equivalent to $\cst_T$ and $\ccost_T$, respectively.

\begin{lemma}\label{Japlegal}
{\em (Japaridze \cite{Japface})} Assume $A$ is a static game, $\Omega$ is a $\wp$-delay of $\Gamma$, and $\Omega$ is a $\wp$-illegal run of $\st A$. Then $\Gamma$ is also a $\wp$-illegal run of $\st A$.
\end{lemma}

\begin{theorem}
The class of static games is closed under $\cst_L$ and $\ccost_L$.
\end{theorem}
\begin{proof}
Since $\ccost_L$ can be expressed through $\cst_L$ and $\neg$, with $\neg$ already known (\cite{Jap03}) to preserve the static property of games, we need only to consider $\cst_L$. In what follows, $A$ is a static game. We want to show that $\cst_L A$ is also static.

Assume $\Gamma$ is a $\wp$-legal run of $\cst_L A$, and $\Omega$ is a $\wp$-delay of $\Gamma$. We need to show that $\Omega$ is also a $\wp$-legal run of $\cst_L A$. Since the legal runs of $\cst_L A$, by Definition \ref{def}, are the same as the legal runs of $\st A$, Lemma \ref{Japlegal} still holds with $\cst_L A$ instead of $\st A$. So, by this lemma with $\cst_L A$ instead of $\st A$, we immediately get that $\Omega$ is a $\wp$-legal run of $\cst_L A$.

Assume $\Gamma$ is a $\wp$-won run of $\cst_L A$, and $\Omega$ is a $\wp$-delay of $\Gamma$. We will show that $\Omega$ is also a $\wp$-won run of $\cst_L A$, thus completing our proof of the promise. If $\Omega$ is a $\neg\wp$-illegal run of $\cst_L A$, then $\Omega$ is won by $\wp$ as promised. Assume that $\Omega$ is not $\neg\wp$-illegal, i.e., $\Omega$ is $\neg\wp$-legal. Then we claim that $\Gamma$ is also $\neg\wp$-legal. First, by Lemma 4.6 of \cite{Jap03}, $\Omega$ is a $\wp$-delay of $\Gamma$ implies that $\Gamma$ is a $\neg\wp$-delay of $\Omega$. Next, if $\Gamma$ is a $\neg\wp$-illegal run of $\cst_L A$, by Lemma \ref{Japlegal} with $\cst_L A$ instead of $\st A$, $\Omega$ is also a $\neg\wp$-illegal run of $\cst_L A$, contrary to our assumption. Hence, $\Gamma$ is $\neg\wp$-legal. On the other hand, since $\Gamma$ is a $\wp$-won run of $\cst_L A$, it is obvious that $\Gamma$ is a $\wp$-legal run of $\cst_L A$. Then, by the previously proven fact, $\Omega$ is also a $\wp$-legal run of $\cst_L A$. Thus, both $\Gamma$ and $\Omega$ are legal runs of $\cst_L A$.
But $\Gamma$ being a legal, $\wp$-won run of $A$ means that, for every (if $\wp=\top$) or some (if $\wp=\bot$) infinite but essentially finite bitstring $v$, $\Gamma^{\preceq v}$ is a $\wp$-won run of $A$. Therefore, as $A$ is static and $\Omega^{\preceq v}$ is obviously a $\wp$-delay of $\Gamma^{\preceq v}$, $\Omega^{\preceq v}$ is also a $\wp$-won run of $A$. Hence $\Omega$ is a $\wp$-won run of $\cst_L A$.
\end{proof}

\begin{theorem}\label{atom}
For any formula $F$, the formulas $\st_T^{\aleph_0} F\rightarrow\st_L^{\aleph_0} F$ and $\st_L^{\aleph_0} F\rightarrow\st_T^{\aleph_0} F$ are uniformly valid.
\end{theorem}

\begin{proof}Our proof here almost literally follows the proof of Theorem 4.1 of \cite{Japface}.

Firstly, we prove the uniform validity of $\st_T^{\aleph_0} F\rightarrow\st_L^{\aleph_0} F$, which means that we should construct an EPM ${\cal M}_1$ such that, for any static game $A$, ${\cal M}_1$ wins $\st_T^{\aleph_0} A\rightarrow\st_L^{\aleph_0} A$, i.e. $\cost_T^{\aleph_0}\neg A\vee\st_L^{\aleph_0} A$. Such an EPM (strategy) ${\cal M}_1$ can be constructed as a machine that repeats the following routine over and over again (possibly infinitely many times). At any step of the strategy, $\Psi$ denotes $\Phi^{1.}$, where $\Phi$ is the then-current position of the play. That is, $\Psi$ is the then-current position in the  $\cost_T^{\aleph_0}\neg A$ component of the overall game.\vspace{2mm}

ROUTINE: Keep granting permission until the adversary makes a move $\alpha$ satisfying the conditions of one of the following two cases, and then act as prescribed in that case.

{\bf Case 1}: $\alpha$ is a move $w.\beta$ in $\cost_T^{\aleph_0}\neg A$, where $w$ is a finite bitstring. Make the same move $w.\beta$ in $\st_L^{\aleph_0} A$.

{\bf Case 2}: $\alpha$ is a move $w.\beta$ in $\st_L^{\aleph_0} A$, where $w$ is a finite bitstring. Make a series of replicative moves in $\cost_T^{\aleph_0}\neg A$ if necessary, so that $w$ becomes a node of the underlying BT-structure of $\langle\Psi\rangle\cost_T^{\aleph_0}\neg A$. Then make the move $w.\beta$ in $\cost_T^{\aleph_0}\neg A$.\vspace{2mm}

Let $\Gamma$ be any run that could be generated by such ${\cal M}_1$. According to the description of ROUTINE, ${\cal M}_1$ (in the role of $\top$) does not make any illegal moves unless its adversary does so first. So, if $\Gamma$ is an illegal run of $\cost_T^{\aleph_0}\neg A\vee\st_L^{\aleph_0} A$, then it is $\bot$-illegal and hence $\top$ is the winner. Suppose now $\Gamma$ is a legal run of $\cost_T^{\aleph_0}\neg A\vee\st_L^{\aleph_0} A$. Let $\Sigma=\Gamma^{1.}$ and $\Pi=\Gamma^{2.}$. In other words, $\Sigma$ is the run that took place in the $\cost_T^{\aleph_0}\neg A$ component, and $\Pi$ is the run that took place in the $\st_L^{\aleph_0} A$ component. If for all infinite but essentially finite bitstrings $v$, $\Pi^{\preceq v}$ is a $\top$-won run of $A$, then $\top$ wins the whole game $\cost_T^{\aleph_0}\neg A\vee\st_L^{\aleph_0} A$ because it wins the component $\st_L^{\aleph_0} A$. Now assume there exists an infinite but essentially finite bitstring $v$ such that $\Pi^{\preceq v}$ is a $\bot$-won run of $A$. From the above strategy we can see that the run taking place in thread $v$ of $\neg A$ is the same as the run taking place in thread $v$ of $A$, with the only difference that $\top$ and $\bot$ are interchanged. That is, $\Sigma^{\preceq v}=\neg\Pi^{\preceq v}$. Therefore, $\Sigma^{\preceq v}$ is a $\top$-won run of $\neg A$, and hence $\Sigma$ is a $\top$-won run of $\cost_T^{\aleph_0}\neg A$, and hence $\Gamma$ a $\top$-won run of the overall game $\cost_T^{\aleph_0}\neg A\vee\st_L^{\aleph_0} A$.

Secondly, we show that the formula $\st_L^{\aleph_0} F\rightarrow\st_T^{\aleph_0} F$ is uniformly valid, meaning that there exists an EPM/strategy ${\cal M}_2$ that wins $\cost_L^{\aleph_0}\neg A\vee\st_T^{\aleph_0} A$ for any static game $A$. Such a strategy ${\cal M}_2$ repeats the following routine over and over again. At any step of the strategy, $\Psi$ denotes $\Phi^{2.}$, where $\Phi$ is the then-current position of the play. In other words, $\Psi$ is the then-current position of the component $\st_T^{\aleph_0} A$. Furthermore, a function $f$ from the leaves $v$ of the underlying BT-structure of $\langle\Psi\rangle\st_T^{\aleph_0} A$ to finite bitstrings $f(v)$ is maintained by ${\cal M}_2$ such that, for any two leaves $v_1\neq v_2$, $f(v_1)$ is not a prefix of $f(v_2)$. At the beginning, i.e. when $\Psi$ is empty, of the play, the empty string $\epsilon$ is the only leaf of the underlying BT-structure of $\langle\Psi\rangle\st_T^{\aleph_0} A$, and the value of $f(\epsilon)$ is initialized to $\epsilon$.\vspace{2mm}

ROUTINE: Keep granting permission until the adversary makes a move $\alpha$ satisfying the conditions of one of the following three cases, and then act as that case prescribes. In what follows, $w$ ranges over finite bitstrings.

{\bf Case 1:} $\alpha$ is a replicative move $w$: in $\st_T^{\aleph_0} A$. Let $v=f(w)$. Then update $f$ by setting $f(w0)=v0, f(w1)=v1$, with the value of $f$ on any other leaves of the underlying BT-structure of $\langle\Psi\rangle\st_T^{\aleph_0} A$ unchanged, and do not make any moves.

{\bf Case 2:} $\alpha$ is a non-replicative move $w.\beta$ in $\st_T^{\aleph_0} A$. Let $u_1,\ldots,u_n$ be all leaves $u$ of the underlying BT-structure of $\langle\Psi\rangle\st_T^{\aleph_0} A$ such that $w$ is a prefix of $u$, and let $v_1=f(u_1),\ldots,v_n=f(u_n)$. Then make the series of  moves $v_1.\beta,\ldots,v_n.\beta$ in $\cost_L^{\aleph_0}\neg A$, leaving the value of $f$ unchanged.

{\bf Case 3:} $\alpha$ is a move $w.\beta$ in $\cost_L^{\aleph_0}\neg A$. First assume that there is a unique leaf $x$ in the underlying BT-structure of $\langle\Psi\rangle\st_T^{\aleph_0} A$ such that $w$ is a proper extension of $f(x)$. Let $v=f(x)$, and $w=vu$ for some nonempty finite bitstring $u$. If there is a ``1" in $u$, then ignore the move $\alpha$, leaving the value of $f$ unchanged and making no moves. If there is no ``1" in $u$, i.e. there are only ``0"s in $u$, then update $f$ by letting $f(x)=w$ without changing the value of $f$ on any other leaves, and make the move $x.\beta$ in $\st_T^{\aleph_0} A$. Now assume that there is no leaf $x$ in the underlying BT-structure of $\langle\Psi\rangle\st_T^{\aleph_0} A$ such that $w$ is a proper extension of $f(x)$. Let $y_1,\ldots,y_n$ (possibly $n=0$) be all leaves $y$ of the underlying BT-structure of $\langle\Psi\rangle\st_T^{\aleph_0} A$ such that $w$ is a prefix of $f(y)$. Then make the series of moves $y_1.\beta,\ldots,y_n.\beta$ in $\st_T^{\aleph_0} A$ and leave the value of $f$ unchanged.\vspace{2mm}

Consider any run $\Gamma$ that could be generated by the above machine ${\cal M}_2$. We may, again, assume that $\Gamma$ is a legal run of $\cost_L^{\aleph_0}\neg A\vee\st_T^{\aleph_0} A$, for otherwise it is $\bot$-illegal and hence $\top$-won. Let $\Sigma=\Gamma^{1.}$ and $\Pi=\Gamma^{2.}$. In other words, $\Sigma$ is the run that took place in $\cost_L^{\aleph_0}\neg A$, and $\Pi$ is the run that took place in $\st_T^{\aleph_0} A$. For a number $i$ such that ROUTINE is iterated at least $i$ times, we use $f_i$ to denote the value of $f$ at the beginning of the $i$'th iteration, and use $\Psi_i$ to denote the position reached by that time in the $\st_T^{\aleph_0} A$ component.

Let $v$ be any infinite but essentially finite bitstring. Suppose that $\Pi^{\preceq v}$ is a $\bot$-won run of $A$ (if there is no such $v$, then $\top$ wins the overall game $\cost_L^{\aleph_0}\neg A\vee\st_T^{\aleph_0} A$, as desired). Let $z$ be an infinite bitstring satisfying the following condition: for any $i$ such that ROUTINE is iterated at least $i$ times, we have that $f_i(v_i)$ is a prefix of $z$, where $v_i$ is the unique prefix of $v$ such that $v_i$ is a leaf of the underlying BT-structure of $\langle\Psi_i\rangle\st_T^{\aleph_0} A$.
From the description of ROUTINE, we see that the following property of $f$ is maintained: for any two finite bitstrings $x_1$ and $x_2$, if $x_1\preceq x_2$, then $f(x_1)\preceq f(x_2)$. For any given $v$ and any $i$ in the previous sense, we have $v_i\preceq v_{i+1}$, and hence $f_i(v_i)\preceq f_{i+1}(v_{i+1})$. Therefore, a $z$ satisfying the above condition indeed exists.

From the description of ROUTINE we can see that what happened in thread $z$ of $\neg A$ is the same as what happened in thread $v$ of $A$ with $\top$ interchanged with $\bot$. Namely, $\Sigma^{\preceq z}=\neg\Pi^{\preceq v}$.  Therefore, $\Sigma^{\preceq z}$ is a $\top$-won run of $\neg A$. All that is left to show is that $z$ is essentially finite. According to the description of ROUTINE, the steps making $z$ different from $v$ could occur in Case 3. But Case 3 could only make $z$ different from $v$ in that $z$ can be obtained by inserting in $v$ some ``0"s between some two ``1"s. Namely, the number of ``1"s in $z$ is the same as that of ``1"s in $v$. So, due to the essential finiteness of $v$, $z$ is also essentially finite. Hence $\Sigma$ is a $\top$-won run of $\cost_L^{\aleph_0}\neg A$, and hence $\Gamma$ a $\top$-won run of $\cost_L^{\aleph_0}\neg A\vee\st_T^{\aleph_0} A$, which ends our proof.
\end{proof}

\begin{lemma}\label{lemma1}
Any formula of the form $\st F\rightarrow\cst_L F$ is uniformly valid.  
\end{lemma}
\begin{proof}
To prove the uniform validity of $\st F\rightarrow\cst_L F$, we should construct an EPM ${\cal M}$ such that, for any static game $A$, ${\cal M}$ wins $\st A\rightarrow\cst_L A$, i.e. $\cost\neg A\vee\cst_L A$. The work of such an EPM (strategy) ${\cal M}$ is very simple. It keeps granting permission, and whenever the adversary makes a move $1.w.\alpha$ for some finite bitstring $w$ and some move $\alpha$, it makes the move $2.w.\alpha$, and vice versa: whenever the adversary makes a move $2.w.\alpha$, it makes the move $1.w.\alpha$.

Consider any run $\Gamma$ generated by ${\cal M}$. It is obvious that ${\cal M}$ never makes illegal moves unless its adversary does so first. Hence we may safely assume that $\Gamma$ is a legal run of $\cost\neg A\vee\cst_L A$. Let $\Sigma=\Gamma^{1.}$ and $\Pi=\Gamma^{2.}$. In other words, $\Sigma$ is the run that took place in the $\cost\neg A$ component, and $\Pi$ is the run that took place in the $\cst_L A$ component.
If for all infinite but essentially finite bitstrings $v$, $\Pi^{\preceq v}$ is a $\top$-won run of $A$, then $\top$ wins the whole game $\cost\neg A\vee\cst_L A$ because it wins the $\cst_L A$ component. Now assume there exists an infinite but essentially finite bitstring $v$ such that $\Pi^{\preceq v}$ is a $\bot$-won run of $A$. From the above strategy we can see that the run took place in thread $v$ of $\neg A$ is the same as the run that took place in thread $v$ of $A$, with the only difference that $\top$ and $\bot$ are interchanged. Namely, $\Sigma^{\preceq v}=\neg\Pi^{\preceq v}$. Therefore, $\Sigma^{\preceq v}$ is a $\top$-won run of $\neg A$, and hence $\Sigma$ is a $\top$-won run of $\cost\neg A$, and hence $\Gamma$ a $\top$-won run of the overall game $\cost\neg A\vee\cst_L A$.
\end{proof}

\begin{lemma}\label{lemma2}
Any formula of the form $\cst_L(E\rightarrow F)\rightarrow (\cst_L E\rightarrow\cst_L F)$ is uniformly valid.  
\end{lemma}

\begin{proof}
To prove the uniform validity of $\cst_L(E\rightarrow F)\rightarrow (\cst_L E\rightarrow\cst_L F)$, we should construct an EPM ${\cal M}$ such that, for any static games $A$ and $B$,  ${\cal M}$ wins $\cst_L(A\rightarrow B)\rightarrow (\cst_L A\rightarrow\cst_L B)$, i.e. $\ccost_L(A\wedge\neg B)\vee(\ccost_L\neg A\vee\cst_L B)$. Such an EPM ${\cal M}$ works as follows. It keeps granting permission. Whenever the adversary makes a move $1.w.1.\alpha$, where $w$ is some finite bitstring and $\alpha$ is some move, it makes a move $2.1.w.\alpha$; whenever the adversary makes a move $1.w.2.\alpha$, it makes a move $2.2.w.\alpha$. And vice versa: whenever the adversary makes a move $2.1.w.\alpha$ for some finite bitstring $w$ and some move $\alpha$, it makes a move $1.w.1.\alpha$; whenever the adversary makes a move $2.2.w.\alpha$, it makes a move $1.w.2.\alpha$.

Consider any run generated by ${\cal M}$ when playing the overall game $\ccost_L(A\wedge\neg B)\vee(\ccost_L\neg A\vee\cst_L B)$.  We may assume that $\Gamma$ is a legal run of the overall game because ${\cal M}$ never makes illegal moves unless its adversary does so first. Let $\Sigma=\Gamma^{1.}$ and $\Pi=\Gamma^{2.}$. Namely, $\Sigma$ is the run that took place in the $\ccost_L(A\wedge\neg B)$ component, and $\Pi$ is the run that took place in the $\ccost_L\neg A\vee\cst_L B$ component. If there exists an infinite but essentially finite bitstring $v$ such that $\Sigma^{\preceq v}$ is a $\top$-won run of $A\wedge\neg B$, then $\top$ is the winner in the $\ccost_L(A\wedge\neg B)$ component, and hence $\top$ wins the overall game. If for every infinite but essentially finite bitstring $v$, $\Sigma^{\preceq v}$ is a $\bot$-won run of $A\wedge\neg B$, then $\bot$ wins at least $A$ or $\neg B$ in thread $v$. But the run that took place in $A$ (resp. $\neg B$) in the thread $v$ of $A\wedge\neg B$ is the same as the run that took place in the thread $v$ of $\neg A$ (resp. $B$) in the $\ccost_L\neg A$ (resp. $\cst_L B$) component, only with $\top$ interchanged with $\bot$. Hence we have that for every infinite but essentially finite bitstring $v$, at least $(\Pi^{1.})^{\preceq v}$ is a $\top$-won run of $\neg A$, or $(\Pi^{2.})^{\preceq v}$ is a $\top$-won run of $B$. This means that $\top$ is the winner in the $\ccost_L\neg A\vee\cst_L B$ component, and hence the winner in the overall game.
\end{proof}

\begin{theorem}\label{formula}
For any formula $F$, the formulas $F_T\rightarrow F_L$ and $F_L\rightarrow F_T$ are uniformly valid. 
\end{theorem}

\begin{proof}
We prove this theorem by induction on the complexity of $F$.

(i) The basis of induction is trivial: when $F$ is an atom $P$, we have $F_T=F_L=P$. It is known (\cite{Japfin}) that affine logic is sound with respect to uniform validity, and that the formula $P\rightarrow P$ is provable in affine logic. So, we have $\uvalid P\rightarrow P$.

(ii) In this and the remaining clauses of this proof, when affine logic proves a formula $A$, we may simply say that $\uvalid A$ for the reason explained in the preceding clause. Assume that $F=\neg E$ for some formula $E$. Now we should show that $\uvalid\neg E_T\rightarrow\neg E_L$ and $\uvalid\neg E_L\rightarrow\neg E_T$. By the induction hypothesis, we have $\uvalid E_T\rightarrow E_L\ (1)$ and $\uvalid E_L\rightarrow E_T\ (2)$. So, by (1) (resp. (2)), $\uvalid(A\rightarrow B)\rightarrow(\neg B\rightarrow\neg A)$ and modus ponens, which was proved in \cite{Japfin} to hold with respect to uniform validity\footnote{Strictly speaking, the sort of formulas for which this fact was proven in \cite{Japfin} is not the same as formulas in our present sense. However, this is irrelevant because the proof of \cite{Japfin} automatically goes through for any class of formulas.}, we have $\uvalid\neg E_L\rightarrow\neg E_T$ (resp. $\uvalid\neg E_T\rightarrow\neg E_L$).

(iii) Assume that $F=E\wedge G$ for some formulas $E$ and $G$. Our goal is to show that $\uvalid E_T\wedge G_T\rightarrow E_L\wedge G_L$ and $\uvalid E_L\wedge G_L\rightarrow E_T\wedge G_T$. By the induction hypothesis, we have $\uvalid E_T\rightarrow E_L\ (1)$, \ \ $\uvalid G_T\rightarrow G_L\ (2)$, \ \ $\uvalid E_L\rightarrow E_T\ (3)$, \ \ $\uvalid G_L\rightarrow G_T\ (4)$. By $(1)$, $\uvalid (A\rightarrow A')\rightarrow ((B\rightarrow B')\rightarrow (A\wedge B\rightarrow A'\wedge B'))$ and modus ponens, we have $\uvalid (G_T\rightarrow G_L)\rightarrow (E_T\wedge G_T\rightarrow E_L\wedge G_L)\ (5)$. Again, by $(2), (5)$, and modus ponens, we have $\uvalid E_T\wedge G_T\rightarrow E_L\wedge G_L$. Similarly, $\uvalid E_L\wedge G_L\rightarrow E_T\wedge G_T$.

(iv) Assume that $F=E\vee G$ for some formulas $E$ and $G$. This case can be proven in a similar way to the preceding clause, with the only difference that in this case we depend on ``$\uvalid (A\rightarrow A')\rightarrow ((B\rightarrow B')\rightarrow (A\vee B\rightarrow A'\vee B'))$" instead of ``$\uvalid (A\rightarrow A')\rightarrow ((B\rightarrow B')\rightarrow (A\wedge B\rightarrow A'\wedge B'))$".

(v) Assume that $F=\cst E$ for some formula $E$. Below we should show that $\uvalid \cst_T E_T\rightarrow \cst_L E_L$ and $\uvalid \cst_L E_L\rightarrow \cst_T E_T$. By the induction hypothesis, we have $\uvalid E_T\rightarrow E_L\ (1)$,\ \ $\uvalid E_L\rightarrow E_T\ (2)$. By the known fact that if $\uvalid A$, then $\uvalid \st A$ (proven in \cite{Jap03}), $(1)$ implies that $\uvalid \st (E_T\rightarrow E_L)\ (3)$. On the other hand, by Lemma \ref{lemma1}, we have $\uvalid \st(E_T\rightarrow E_L)\rightarrow \cst_L(E_T\rightarrow E_L)\ (4)$. So, by $(3),(4)$, and modus ponens, we have $\uvalid \cst_L(E_T\rightarrow E_L)\ (5)$. Next, by $(5)$, Lemma \ref{lemma2} and modus ponens, we obtain $\uvalid \cst_L E_T\rightarrow \cst_L E_L\ (6)$. In addition, by Theorem \ref{atom}, $\uvalid \cst_T E_T\rightarrow \cst_L E_T\ (7)$. Finally, by $(6),(7)$, $\uvalid (A\rightarrow B)\rightarrow ((B\rightarrow C)\rightarrow(A\rightarrow C))$ and modus ponens, we get $\uvalid \cst_T E_T\rightarrow \cst_L E_L$ as one of our desired results. In a similar way, we can show that $\uvalid \cst_L E_L\rightarrow \cst_T E_T$.

(vi) Assume that $F=\ccost E$ for some formula $E$. By the induction hypothesis, $\uvalid E_T\rightarrow E_L\ (1)$. By $(1)$ and $\uvalid (A\rightarrow B)\rightarrow(\neg B\rightarrow\neg A)$, we have $\uvalid \neg E_L\rightarrow\neg E_T\ (2)$. Then, from $(2)$, as in the preceding clause, we get $\uvalid \cst_L\neg E_L\rightarrow\cst_T\neg E_T\ (3)$, i.e. $\uvalid \neg\ccost_L E_L\rightarrow\neg\ccost_T E_T\ (4)$. Again, by $(4)$ and $\uvalid (A\rightarrow B)\rightarrow(\neg B\rightarrow\neg A)$, we get one of the desired results: $\uvalid \ccost_T E_T\rightarrow \ccost_L E_L$. Similarly, we have $\uvalid \ccost_L E_L\rightarrow \ccost_T E_T$.

(vii) Assume that $F=\st E$ for some formula $E$. By the induction hypothesis, we have $\uvalid E_T\rightarrow E_L\ (1)$, $\uvalid E_L\rightarrow E_T\ (2)$. Then, by $(1)$ and the known fact (\cite{Jap03}) that $\uvalid A$ implies $\uvalid \st A$, we get $\uvalid\st(E_T\rightarrow E_L)\ (3)$. But it is known (\cite{Japfin}) that, for any formulas $A$ and $B$ in affine logic, $\uvalid\st(A\rightarrow B)\rightarrow(\st A\rightarrow\st B)\ (4)$.  So, by (3),(4), and modus ponens, we have $\uvalid\st E_T\rightarrow\st E_L$. Similarly, we have $\uvalid\st E_L\rightarrow\st E_T$.

(viii) Assume that $F=\cost E$ for some formula $E$. By the induction hypothesis and clause (ii), we have $\uvalid\neg E_L\rightarrow\neg E_T\ (1)$. Then, from (1), as in the preceding clause, we get $\uvalid\st\neg E_L\rightarrow\st\neg E_T$, i.e. $\uvalid\neg\cost E_L\rightarrow\neg\cost E_T\ (2)$. Finally, by (2), $\uvalid(A\rightarrow B)\rightarrow(\neg B\rightarrow\neg A)$ and modus ponens, we get $\uvalid\neg\neg\cost E_T\rightarrow\neg\neg\cost E_L$, i.e. $\uvalid\cost E_T\rightarrow\cost E_L$. In a similar way, we get that $\uvalid\cost E_L\rightarrow\cost E_T$.
\end{proof}

\begin{corollary}\label{equivalent}
For any formula $F$, $F_T$ is uniformly valid iff so is $F_L$.
\end{corollary}
\begin{proof}
Immediately form Theorem \ref{formula} and the fact (\cite{Japfin}) that uniform validity is closed under modus ponens.
\end{proof}

In view of Corollary \ref{equivalent}, from now on, when studying the fragments of CoL involving $\cst$ and $\ccost$, we can safely exclusively focus on the new version of $\cst$ and $\ccost$. So, let us agree that, for the rest of the paper, $\cst$ and $\ccost$ always mean $\cst_L$ and $\ccost_L$, respectively.

\section{The soundness of {\bf CL15} with countable branching recurrence}
To prove the soundness of {\bf CL15}$(\st^{\aleph_0})$, we first need to extend the earlier-described semantics from formulas to cirquents. In this section, unless otherwise specified, by a ``formula'' we mean a $(\neg,\wedge,\vee,\cst,\ccost)$-formula.

Let $\Gamma$ be a run, $a$ be a positive integer, and $\vec{x}=x_1,\ldots,x_n$ be a nonempty sequence of $n$ infinite bitstrings. The notation
\begin{center}
$\Gamma^{\preceq a;\vec{x}}$
\end{center}
will be used to indicate the result of deleting from $\Gamma$ all moves (together with their labels) except those that look like $a;u_1,\ldots,u_n.\beta$ for some move $\beta$ and some finite initial segments $u_1,\ldots,u_n$ of $x_1,\ldots,x_n$, respectively, and
then further deleting the prefix ``$a;u_1,\ldots,u_n.$" from such moves.
For instance, $\langle\bot 1;100,11.\alpha,\ \top 1;01,100.\beta,\ \bot 1;1,1.\gamma,\ \bot 2;100,111.\delta\rangle^{1;100\ldots,111\ldots}=\langle\bot\alpha,\ \bot\gamma\rangle$.


\begin{definition}\label{def1}
Let $^*$ be an interpretation, and $C=(\langle F_1,\ldots,F_k\rangle,\langle U_1,\ldots,U_m\rangle,\langle O_1,\ldots,O_n\rangle)$ be a cirquent. Then $C^*$ is the game defined as follows, where $\Gamma$ is an arbitrary run and $\Omega$ is any legal run of $C^*$.\vspace{2mm}\\
{\bf (i)} $\Gamma\in {\bf Lr}^{C^*}$ iff the following two conditions are satisfied:
\begin{itemize}
\item Every move of $\Gamma$ looks like $a;\vec{u}.\alpha$, where $\alpha$ is some move, $a\in\{1,\ldots,k\}$, and $\vec{u}=u_1,\ldots,u_n$ is a sequence of $n$ finite bitstrings such that, whenever an overgroup $O_j$ $(1\leq j\leq n)$ does not contain the oformula $F_a$, $u_j=\epsilon$.
\item For every $a\in\{1,\ldots,k\}$ and every sequence $\vec{x}$ of $n$ infinite bitstrings, $\Gamma^{\preceq a;\vec{x}}\in {\bf Lr}^{F_a^{*}}$.
\end{itemize}
{\bf (ii)} ${\bf Wn}^{C^*}\langle\Omega\rangle=\top$ iff, for every $i\in\{1,\ldots,m\}$ and every sequence $\vec{x}$ of $n$ infinite but essentially finite bitstrings, there is an $a\in\{1,\ldots,k\}$ such that the undergroup $U_i$ contains the oformula $F_a$ and ${\bf Wn}^{F_a^*}\langle\Omega^{\preceq a;\vec{x}}\rangle=\top$.
\end{definition}

\begin{remark}\label{feb13a}
Intuitively, any legal run $\Omega$ of $C^*$ consists of parallel plays of countably infinite copies/threads of each of the games $F_{a}^{*}$ ($1\leq a\leq k$). To every sequence $\vec{x}$ of $n$ infinite but essentially finite bitstrings corresponds a thread of $F_a^*$, and $\Omega^{\preceq a;\vec{x}}$ is the run played in that thread. We shall simply say {\bf the thread $\vec{x}$} of $F_a^*$ to mean the copy of $F_a^*$ which corresponds to the sequence $\vec{x}$. Now, consider a given undergroup $U_i$. $\top$ is the winner in $U_i$ iff, for every sequence $\vec{x}$ of $n$ infinite but essentially finite bitstrings, there is an oformula $F_a$ in $U_i$ such that $\Omega^{\preceq a;\vec{x}}$ is won by $\pp$. Finally, $\pp$ wins the overall game $C^*$ iff it wins in all undergroups of $C$. In fact, overgroups can be seen as generalized $\cst$s, with the only main difference that the former can be shared by several oformulas; undergroups can be seen as generalized disjunctions, with the only main difference that the former may have shared arguments with other undergroups.

\end{remark}

We say that a cirquent $C$ is
{\bf uniformly valid} iff there is a machine $\cal M$, called a {\bf uniform solution} of $C$, such that, for every interpretation $^*$, $\cal M$ wins $C^*$.\vspace{3mm}

\begin{lemma}\label{lemma3}
The formula $\cst P\rightarrow P$ is uniformly valid.
\end{lemma}
\begin{proof}
This is one exception where we prefer to deal with the old version $\cst_T$ of $\cst$. Our goal is to show that there exists an EPM ${\cal M}$ such that, for any static game $A$, ${\cal M}$ wins $\cst_T A\rightarrow A$, i.e. $\ccost_T\neg A\vee A$. Such an EPM ${\cal M}$ works as follows. It never makes any replicative moves in the left component. Whenever the environment makes a move $1.\epsilon.\alpha$ for some move $\alpha$, it makes the move $2.\alpha$; and whenever the environment makes a move $2.\beta$ for some move $\beta$, it makes the move $1.\epsilon.\beta$.

Consider any run $\Gamma$ generated by ${\cal M}$. As earlier, we assume that $\Gamma$ is a legal run of the overall game. Let $\Sigma=\Gamma^{1.}$ and $\Pi=\Gamma^{2.}$. That is, $\Sigma$ is the run that took place in the $\ccost_T\neg A$ component, and $\Pi$ is the run that took place in the $A$ component. If there is an infinite but essentially finite bitstring $v$ such that $\Sigma^{\preceq v}$ is a $\top$-won run of $\neg A$, then ${\cal M}$ wins the $\ccost_T\neg A$ component, and hence wins the overall game. Now assume that, for every infinite but essentially finite bitstring $v$, $\Sigma^{\preceq v}$ is a $\bot$-won run of $\neg A$. But from the description of the work of ${\cal M}$, one can easily see that $\Sigma^{\preceq v}=\neg\Pi$ for every such $v$. Therefore, $\Pi$ is a $\top$-won run of $A$, and hence $\Gamma$ is won by ${\cal M}$.
\end{proof}

It should be acknowledged that the following proofs in the present section very closely follow
the proofs of \cite{Japtam}.
\begin{lemma}\label{apr14a}
There is an effective function $f$ from machines to machines such that, for every machine ${\cal M}$, formula $F$ and interpretation $^*$, if ${\cal M}$ wins $\cst F^*$, then $f({\cal M})$ wins $F^*$.
\end{lemma}

\begin{proof}
Lemma \ref{lemma3} almost immediately implies that there is a machine ${\cal N}_0$ such that ${\cal N}_0$ wins $\cst F^*\rightarrow F^*$ for any formula $F$ and interpretation $^*$. Furthermore, by Proposition 21.3 of \cite{Jap03}, there is an effective procedure that, for any pair $({\cal N},{\cal M})$ of machines, returns a machine $h({\cal N},{\cal M})$ such that, for any static games $A$ and $B$, if ${\cal N}$ wins $A\rightarrow B$ and ${\cal M}$ wins $A$, then $h({\cal N},{\cal M})$ wins $B$. So, let $f({\cal M})$ be the function satisfying  $f({\cal M})=h({\cal N}_0,{\cal M})$. Then $f({\cal M})$ wins $F^*$.
\end{proof}

\begin{lemma}\label{apr14b}
There is an effective function $g$ from machines to machines such that, for every machine ${\cal M}$, formula $F$ and interpretation $^*$, if ${\cal M}$ wins $(F^{\clubsuit})^*$, then $g({\cal M})$ wins $F^*$.
\end{lemma}

\begin{proof}
Every legal move of $(F^{\clubsuit})^*$ looks like $1;w.\alpha$ for some finite bitstring $w$ and move $\alpha$, while the corresponding legal move of $(\cst F)^*$ simply looks like $w.\alpha$, and vice versa. Consider an arbitrary EPM ${\cal M}$ and an arbitrary  interpretation $^*$. Below we  show the existence of an effective function $f$ such that, if ${\cal M}$ wins $(F^{\clubsuit})^*$, then (the strategy) $f({\cal M})$ wins $(\cst F)^*$.

We construct an EPM $f({\cal M})$ that plays $(\cst F)^*$ by simulating and mimicking a play of $(F^{\clubsuit})^*$ (called the {\bf imaginary play}) by ${\cal M}$ as follows. Throughout simulation, $f({\cal M})$ grants permission whenever the simulated ${\cal M}$ does so, and feeds its environment's response---in a slightly modified form described below---back to the simulated $\cal M$ as the response of ${\cal M}$'s imaginary adversary (this detail of simulation will no longer be explicitly mentioned later in similar situations).  Whenever the environment makes a move $w.\alpha$ for some finite bitstring $w$ and move $\alpha$, $f({\cal M})$  translates it as the move $1;w.\alpha$ made by the imaginary adversary of ${\cal M}$, and ``vice versa":
whenever the simulated ${\cal M}$ makes a move $1;w.\alpha$ for some finite bitstring $w$ and move $\alpha$ in the imaginary play of $(F^{\clubsuit})^*$, $f({\cal M})$ translates it as its own move $w.\alpha$ in the real play of $(\cst F)^*$.
The effect achieved by $f({\cal M})$'s strategy can be summarized by saying that it synchronizes every thread $x$ of $F^*$ in the real play of $(\cst F)^*$ with the ``same thread" $x$ of $F^*$ in the imaginary play of $(F^{\clubsuit})^*$.

Let $\Gamma$ be an arbitrary run generated by $f({\cal M})$, and $\Omega$ be the corresponding run in the imaginary play of $(F^{\clubsuit})^*$ by ${\cal M}$. From our description of  $f({\cal M})$ it is clear that the latter never makes illegal moves unless its environment or the simulated ${\cal M}$ does so first. Hence we may safely assume that $\Gamma$ is a legal run of $(\cst F)^*$ and $\Omega$ is a legal run of $(F^{\clubsuit})^*$, for otherwise either $\Gamma$ is a $\bot$-illegal run of $(\cst F)^*$ and thus $f({\cal M})$ is an automatic winner in $(\cst F)^*$, or $\Omega$ is a $\top$-illegal run of $(F^{\clubsuit})^*$ and thus ${\cal M}$ does not win $(F^{\clubsuit})^*$. Now, it is not hard to see that, for any infinite but essentially finite bitstring $x$, we have $\Gamma^{\preceq x}=\Omega^{\preceq 1;x}$. Therefore, $f({\cal M})$ wins $(\cst F)^*$ as long as ${\cal M}$ wins $(F^{\clubsuit})^*$.

Finally, in view of Lemma \ref{apr14a}, the existence of function $g$ satisfying the promise of the present lemma is obviously guaranteed.
\end{proof}

A rule of {\bf CL15}$(\cst)$ (other than Axiom) is said to be {\bf uniform-constructively sound} iff there is an effective procedure that takes any instance $(A,B)$ (i.e. a particular premise-conclusion pair) of the rule, any machine ${\cal M}_A$ and returns a machine ${\cal M}_B$ such that, for any interpretation $^*$, whenever ${\cal M}_A$ wins $A^*$, ${\cal M}_B$ wins $B^*$. Axiom is uniform-constructively sound iff there is an effective procedure that takes any instance $B$ of (the ``conclusion" of) Axiom and returns a uniform solution ${\cal M}_B$ of $B$.

\begin{theorem}\label{mainth1}
All rules  of {\bf CL15}$(\cst)$ are uniform-constructively sound.
\end{theorem}

\begin{proof}
In what follows, $A$ is the premise of an arbitrary instance of a given rule of {\bf CL15}$(\cst)$, and $B$ is the corresponding conclusion, except the case of Axiom where we only have $B$. We will prove that each rule of {\bf CL15}$(\cst)$ is uniform-constructively sound by showing that  an EPM  ${\cal M}_B$ can be constructed effectively from an arbitrary EPM (or BMEPM in some cases) ${\cal M}_A$ such that, for whatever interpretation $^*$,
whenever ${\cal M}_A$ wins $A^{\ast}$,  ${\cal M}_B$ wins $B^{\ast}$. Since an interpretation $^{\ast}$ is never relevant in such proofs, we may safely omit it, writing simply $A$ instead of $A^{\ast}$ to represent a game. Next, in all cases the assumption that ${\cal M}_A$ wins $A$ will be implicitly made, even though it should be pointed out that the construction of ${\cal M}_B$ never depends on this assumption. Correspondingly, it will be assumed that ${\cal M}_A$ never makes illegal moves. Further, as in the proof of Lemma \ref{apr14b}, we shall always implicitly assume that ${\cal M}_B$'s adversary never makes illegal moves either. To summarize, when analyzing ${\cal M}_B$, ${\cal M}_A$ and the games they play, we safely pretend that illegal runs never occur.\vspace{2mm}

{\bf (1)}\ Assume that $B$ is an axiom with $2n$ oformulas. An EPM ${\cal M}_B$ that wins $B$ can be constructed as follows. It keeps granting permission. Whenever the environment makes a move $a;\vec{w}.\alpha$, where $1\leq a\leq 2n$ and $\vec{w}$ is a sequence of $n$ finite bitstrings, ${\cal M}_B$ makes the move $b;\vec{w}.\alpha$, where $b=a+1$ if $a$ is odd, and $b=a-1$ if $a$ is even. Then, for any run $\Gamma$ of $B$ generated by ${\cal M}_B$ and any sequence $\vec{x}$ of $n$ infinite but essentially finite bitstrings, we have $\Gamma ^{\preceq a;\vec{x}}=\neg\Gamma^{\preceq b;\vec{x}}$. It is obvious that $\Gamma$ is a $\top$-won run of $B$, so that ${\cal M}_B$ wins $B$. \vspace{2mm}

{\bf (2)}\ Assume that $B$ follows from $A$ by Overgroup Exchange, where the $i$'th ($i\geq 1$) and the $(i+1)$'th overgroups of $A$ have been swapped when obtaining $B$ from $A$. The EPM ${\cal M}_B$ works by simulating and mimicking ${\cal M}_A$ as follows. Let $n$ be the number of overgroups of either cirquent, and $a$ be a positive integer not exceeding the number of oformulas of either cirquent. For any move (by either player) $a;\vec{w_1},u_1,u_2,\vec{w_2}.\alpha$  in the real play of $B$,  where $\vec{w_1}$ and $\vec{w_2}$ are any sequences of $i-1$ and $n-i-1$ finite bitstrings, respectively, and $u_1,u_2$ are two finite bitstrings,  ${\cal M}_B$ translates it as the move $a;\vec{w_1},u_2,u_1,\vec{w_2}.\alpha$ (by the same player) in the imaginary play of $A$, and vice versa, with all other moves not reinterpreted.

Let $\Gamma$ be any run of $B$ generated by ${\cal M}_B$, and $\Omega$ be the corresponding run generated by ${\cal M}_A$ in the imaginary play of $A$. It is obvious that, for any sequence $\vec{x}$ of $n$ infinite but essentially finite bitstrings, $\Gamma^{\preceq a;\vec{x}}=\Omega^{\preceq a;\vec{y}}$, where $\vec{y}$ is the result of swapping in $\vec{x}$ the $i$'th and $(i+1)$'th bitstrings. Hence ${\cal M}_B$ wins $B$ because ${\cal M}_A$ wins $A$.

In the case of Oformula Exchange, a similar method can be used to construct ${\cal M}_B$, with the only difference that the reinterpreted objects are the occurrences of two adjacent oformulas rather than the occurrences of two adjacent overgroups.

As for Undergroup Exchange, its conclusion, as a game, is the same as its premise. So, the machine ${\cal M}_B={\cal M}_A$ does the job.\vspace{2mm}

In the subsequent clauses, as in the present one, without any further indication, $\Gamma$ will stand for an arbitrary run of $B$ generated by ${\cal M}_B$, and $\Omega$ will stand for the run of $A$ generated by the simulated machine ${\cal M}_A$ in the corresponding scenario. \vspace{2mm}

{\bf (3)}\ Assume $B$ is obtained from $A$ by Weakening. If no oformula of $B$ was deleted when moving from $B$ to $A$, then ${\cal M}_B$ works exactly as ${\cal M}_A$ does and succeeds, because every $\top$-won run of $A$ is also a $\top$-won run of $B$ (but not necessarily vice versa). If, when moving from $B$ to $A$, an oformula $F_a$ of $B$ was deleted, then ${\cal M}_B$ can be constructed as a machine that works by simulating and mimicking ${\cal M}_A$. What ${\cal M}_B$ needs to do during its work is to ignore the moves within $F_a$, and play exactly as ${\cal M}_A$ does in all other oformulas. Again, it is obvious that every $\top$-won run of $A$ is also a $\top$-won run of $B$, which means that ${\cal M}_B$ wins $B$ as long as ${\cal M}_A$ wins $A$.\vspace{2mm}

{\bf (4)}\ Since Exchange has already been proven to be uniform-constructively sound, in this and the remaining clauses of the present proof, we may safely assume that the oformulas and overgroups affected by a rule are at the end of the corresponding lists of objects of the corresponding cirquents.

Assume $B$ follows from $A$ by Contraction, and the contracted oformula $\ccost F$ is at the end of the list of oformulas of $B$. Let $a$ be the number of oformulas of $B$, and let $b=a+1$. Thus, the $a$'th oformula of $B$ is $\ccost F$, and the $a$'th and $b$'th oformulas of $A$ are $\ccost F$ and $\ccost F$. Let $n$ be the number of overgroups in either cirquent.
In this case, we assume that ${\cal M}_A$ is a BMEPM rather than an EPM.
As always, we let ${\cal M}_B$ be an EPM that works by simulating and mimicking ${\cal M}_A$. Namely, let $\vec{w}$ be any sequence of $n$ finite bitstrings. If the moves take place within the oformulas other than $\ccost F$, then nothing should be reinterpreted. If the moves take place in $\ccost F$, then we have:

\begin{itemize}
  \item For any move $a;\vec{w}.0u.\alpha$ made by the environment in the real play of $B$, ${\cal M}_B$ translates it as the move $a;\vec{w}.u.\alpha$ by the imaginary adversary of ${\cal M}_A$ in the play of $A$; whenever the simulated ${\cal M}_A$ makes a move $a;\vec{w}.u.\alpha$ in the imaginary play of $A$, ${\cal M}_B$ makes the move $a;\vec{w}.0u.\alpha$ in the real play of $B$.
  \item For any move $a;\vec{w}.1u.\alpha$ made by the environment in the real play of $B$, ${\cal M}_B$ translates it as the move $b;\vec{w}.u.\alpha$ by the imaginary adversary of ${\cal M}_A$ in the play of $A$; whenever the simulated ${\cal M}_A$ makes a move $b;\vec{w}.u.\alpha$ in the imaginary play of $A$, ${\cal M}_B$ makes the move $a;\vec{w}.1u.\alpha$ in the real play of $B$.
  \item If the environment makes a move $a;\vec{w}.\epsilon.\alpha$ in the real play of $B$, ${\cal M}_B$ translates it as a block of the two moves $a;\vec{w}.\epsilon.\alpha$ and $b;\vec{w}.\epsilon.\alpha$ by the imaginary adversary of ${\cal M}_A$ in the play of $A$, and vice versa.
\end{itemize}

Note that if ${\cal M}_A$ makes a block of several moves at once (because it is a BMEPM), ${\cal M}_B$ still works as described above, with the only difference that it will correspondingly make several consecutive moves in the real play, rather than only one move. In the remaining clauses of the present proof, whenever ${\cal M}_A$ is assumed to be a BMEPM, for simplicity we may assume that it never makes more than one move at once. For, otherwise, a block of several moves made by ${\cal M}_A$ at once will be translated through several consecutive moves by ${\cal M}_B$ as noted above.

Below we show that ${\cal M}_B$ wins $B$, i.e., ${\cal M}_B$ is the winner in every undergroup of $B$. Let $U_i^{B}$ be any $i$'th undergroup of $B$ and $U_i^{A}$ be the corresponding $i$'th undergroup of $A$, and let $\vec{x}$ be any sequence of $n$ infinite but essentially finite bitstrings. Since ${\cal M_A}$ wins $A$, $U_i^{A}$ is won by ${\cal M}_A$. So, for the sequence $\vec{x}$, there is an oformula $F_j$ ($1\leq j\leq b$) in $U_i^{A}$ such that $\Omega^{\preceq j;\vec{x}}$ is a $\top$-won run of $F_j$.
Next, if such $F_j$ is not one of the two contracted oformulas $\ccost F$ and $\ccost F$, then, for $\vec{x}$, the corresponding oformula $F_j$ of $B$ is won by ${\cal M}_B$, i.e. $\Gamma^{\preceq j;\vec{x}}$ is a $\top$-won run of $F_j$, because ${\cal M}_B$ plays in the thread $\vec{x}$ of $F_j$ exactly as ${\cal M}_A$ does. This means that $U_i^B$ is won by ${\cal M}_B$. If $F_j$ is one of the two contracted oformulas $\ccost F$ and $\ccost F$, below let us assume that $F_j$ is the right $\ccost F$, with the case of the left $\ccost F$ being similar. Then there is an infinite but essentially finite bitstring $w$ such that the thread $w$ of $F$ within the thread $\vec{x}$ of the right $\ccost F$ is won by ${\cal M}_A$, i.e. $(\Omega^{\preceq j;\vec{x}})^{\preceq w}$ is a $\top$-won run of $F$. But, according to the above description, ${\cal M}_B$ plays in the thread $1w$ of $F$ within the thread $\vec{x}$ of $\ccost F$ in $B$ exactly as ${\cal M}_A$ plays in the thread $w$ of $F$ within the thread $\vec{x}$ of the right $\ccost F$ in $A$, i.e. $(\Gamma^{\preceq j;\vec{x}})^{\preceq 1w}=(\Omega^{\preceq j;\vec{x}})^{\preceq w}$. Therefore, $(\Gamma^{\preceq j;\vec{x}})^{\preceq 1w}$ is a $\top$-won run of $F$, which means that $\Gamma^{\preceq j;\vec{x}}$ is a $\top$-won run of $\ccost F$ in $B$, and hence the $\ccost F$-containing undergroup $U_i^{B}$ is won by ${\cal M}_B$.\vspace{2mm}

{\em Remark}\hspace{1pt}:\  In the remaining clauses, just as in the preceding one, when talking about playing, winning, etc. in $A$ (resp. $B$) or any of its components, it is to be understood in the context of $\Omega$ (resp. $\Gamma$). Furthermore, if $A$ and $B$ have the same number $n$ of overgroups, then the context will additionally include some arbitrary but fixed sequence $\vec{x}$ of $n$ infinite but essentially finite bitstrings.\vspace{2mm}

{\bf (5)}\ Undergroup Duplication does not modify the game associated with the cirquent, so we only need to consider Overgroup Duplication.

Assume $B$ is obtained from $A$ by Overgroup Duplication. We assume that the duplicated overgroup is at the end of the list of overgroups of $A$. Let $n+1$ be the number of overgroups of $A$. Thus, every legal move of $A$ (resp. $B$) looks like $a;\vec{w},u.\alpha$ (resp. $a;\vec{w},u_1,u_2.\alpha$), where $a$ is a positive integer not exceeding the number of oformulas of $A$, $\vec{w}$ is a sequence of $n$ finite bitstrings, and $u,u_1,u_2$ are finite bitstrings.

Let $x$ and $y$ be any two---finite or infinite---bitstrings, a bitstring $z$ is a {\bf fusion} of $x$ and $y$ iff $z$ is a shortest bitstring satisfying that, for any natural numbers $i,j$ such that $x$ has at least $i$ bits and $y$ has at least $j$ bits, we have: (1) the $(2i-1)$'th bit of $z$ exists and it is the $i$'th bit of $x$; (2) the $(2j)$'th bit of $z$ exists and it is the $j$'th bit of $y$. Here and later the count of bits starts from $1$, and goes from left to right. For instance, if $x=001$ and $y=110$, then they have only one fusion $z=010110$; if $x=01$ and $y=110$, then they have two fusions $z_1=011100$, $z_2=011110$. Note that when both $x$ and $y$ are infinite, they have only one fusion. The {\bf defusion} of a bitstring $z$ is the pair $(x,y)$ where $x$ (resp. $y$) is the result of deleting from $z$ all bits except those that are found in odd (resp. even) positions. For instance, the defusion of $100110101$ is $(10111,0100)$. It is obvious that if $x$ and $y$ are infinite but essentially finite bitstrings, then their unique fusion $z$ is also essentially finite, and vice versa.

In the present case, we assume that ${\cal M}_A$ is a BMEPM.
As before, ${\cal M}_B$ works by simulating ${\cal M}_A$. Whenever ${\cal M}_A$ makes a move $a;\vec{w},u.\alpha$ in $A$,  ${\cal M}_B$ makes the move $a;\vec{w},u_1,u_2.\alpha$ in the real play of $B$, where $(u_1,u_2)$ is the defusion of $u$. And whenever the environment makes a move $a;\vec{w},u_1,u_2.\alpha$ in the real play of $B$, ${\cal M}_B$ translates it as a block of ${\cal M}_A$'s imaginary adversary's moves $a;\vec{w},v_1.\alpha, \ldots, a;\vec{w},v_k.\alpha$ in $B$, where $v_1, \ldots, v_k$ are all the fusions of $u_1$ and $u_2$.

For every oformula $F_a$ of either cirquent, every sequence $\vec y$ of $n$ infinite but essentially finite bitstrings and any infinite but essentially finite bitstrings $x_1$ and $x_2$, we have $\Gamma^{\preceq a;\vec{y},x_1,x_2}=\Omega^{\preceq a;\vec{y},x}$, where $x$ is the fusion of $x_1$ and $x_2$. So it is obvious that ${\cal M}_B$ wins $B$ as long as ${\cal M}_A$ wins $A$.\vspace{2mm}

{\bf (6)}\ Assume $B$ follows from $A$ by Merging. Let us assume that $A$ has $n+2$ overgroups, and $B$ is the result of merging in $A$ the two adjacent overgroups $O_{n+1}$ and $O_{n+2}$. Then every legal move of $A$ (resp. $B$) looks like $a;\vec{w},u_1,u_2.\alpha$ (resp. $a;\vec{w},u.\alpha$), where $a$ is a positive integer not exceeding the number of oformulas in either cirquent, $\vec{w}$ is a sequence of $n$ finite bitstrings, and $u,u_1,u_2$ are finite bitstrings. We still assume that ${\cal M}_A$ is a BMEPM.
The EPM ${\cal M}_B$ works as follows.

If the $a$'th oformula of $A$ is neither in $O_{n+1}$ nor in $O_{n+2}$, then ${\cal M}_B$ interprets every move $a;\vec{w},\epsilon,\epsilon.\alpha$ made by ${\cal M}_A$ in the imaginary play of $A$ as the move $a;\vec{w},\epsilon.\alpha$ in the real play of $B$, and vice versa.

If the $a$'th oformula of $A$ is in $O_{n+1}$ but not in $O_{n+2}$, ${\cal M}_B$ interprets every move $a;\vec{w},u,\epsilon.\alpha$ made by ${\cal M}_A$ in the imaginary play of $A$ as the move $a;\vec{w},u.\alpha$ in the real play of $B$, and vice versa. Namely, ${\cal M}_B$ interprets every move $a;\vec{w},u.\alpha$ by its environment in the real play of $B$ as the move $a;\vec{w},u,\epsilon.\alpha$ by the imaginary adversary of ${\cal M}_A$ in the play of $A$.

The case of the $a$'th oformula of $A$ being in $O_{n+2}$ but not in $O_{n+1}$ is similar.

Finally, suppose that the $a$'th oformula of $A$ is in both $O_{n+1}$ and $O_{n+2}$. Whenever the environment makes a move $a;\vec{w},u.\alpha$ in the real play of $B$, ${\cal M}_B$ translates it as the move $a;\vec{w},u_1,u_2.\alpha$ by the imaginary adversary of ${\cal M}_A$ in the play of $A$, where $(u_1,u_2)$ is the defusion of $u$. Next, whenever ${\cal M}_A$ makes a move $a;\vec{w},u_1,u_2.\alpha$ in the imaginary play of $A$, ${\cal M}_B$ translates it as a series of moves $a;\vec{w},v_1.\alpha, \ldots, a;\vec{w},v_k.\alpha$ in the real play of $B$, where $v_1, \ldots, v_k$ are all the fusions of $u_1$ and $u_2$.

For every oformula $F_a$ of either cirquent, every sequence $\vec y$ of $n$ infinite but essentially finite bitstrings and any infinite but essentially finite bitstring $x$, we have $\Gamma^{\preceq a;\vec{y},x}=\Omega^{\preceq a;\vec{y},x_1,x_2}$, where $x_1,x_2$ are infinite but essentially finite bitstrings satisfying that $x_1=x$ (when $F_a$ is contained in $O_{n+1}$ but not $O_{n+2}$), or $x_2=x$ (when $F_a$ is contained in $O_{n+2}$ but not $O_{n+1}$), or $(x_1,x_2)$ is the defusion of $x$ (when $F_a$ is contained in both $O_{n+1}$ and $O_{n+2}$, or is contained in neither of them). So it is obvious that ${\cal M}_B$ wins $B$ as long as ${\cal M}_A$ wins $A$.\vspace{2mm}

{\bf (7)}\ In this and the remaining clauses of the present proof, we will limit our descriptions to what moves ${\cal M}_B$ needs to properly reinterpreted and how, with any unmentioned sorts of moves implicitly assumed to remain unchanged.

Assume $B$ is obtained from $A$ by Disjunction Introduction. Let us assume that the last ($a$'th) oformula of $B$ is $E\vee F$, and the last two ($a$'th and $b$'th, where $b=a+1$) oformulas of $A$ are $E$ and $F$. As always, ${\cal M}_B$ reinterprets every move $a;\vec{w}.\alpha$ (resp. $b;\vec{w}.\alpha$) by either player in the imaginary play of $A$ as the move $a;\vec{w}.1.\alpha$ (resp. $a;\vec{w}.2.\alpha$) by the same player in the real play of $B$, and vice versa.

Consider any undergroup $U_i^{B}$ of $B$, and let $U_i^{A}$ be the corresponding undergroup of $A$. As before, ${\cal M}_A$'s winning $A$ means that $U_i^{A}$ is won by ${\cal M}_A$, which, in turn, means that there is an oformula $G$ in $U_i^{A}$ that is won by ${\cal M}_A$. If $G$ is neither $E$ nor $F$, then the oformula $G$ of $B$ is also won by ${\cal M}_B$, because ${\cal M}_B$ plays in $G$ exactly as ${\cal M}_A$ does. Hence $U_i^{B}$ is won by ${\cal M}_B$. If $G$ is $E$, then its being $\top$-won means that ${\cal M}_B$ wins the $E$ component of $E\vee F$, because ${\cal M}_B$ plays in the $E$ component of $E\vee F$ exactly as ${\cal M}_A$ plays in $E$. Therefore, $E\vee F$ is won by ${\cal M}_B$, and hence so is the $E\vee F$-containing undergroup $U_i^{B}$. The case of $G$ being $F$ is similar.\vspace{2mm}

{\bf (8)}\ Assume $B$ follows from $A$ by Conjunction Introduction. We also assume that the last ($a$'th) oformula of $B$ is $E\wedge F$, and the last two ($a$'th and $b$'th, where $b=a+1$) oformulas of $A$ are $E$ and $F$. As the case of Disjunction Introduction, ${\cal M}_B$ reinterprets every move $a;\vec{w}.\alpha$ (resp. $b;\vec{w}.\alpha$) by either player in the imaginary play of $A$ as the move $a;\vec{w}.1.\alpha$ (resp. $a;\vec{w}.2.\alpha$) by the same player in the real play of $B$, and vice versa.

Let $U_i$ be any undergroup of $B$. If $U_i$ does not contain $E\wedge F$, then the corresponding undergroup $V_i$ of $A$ contains neither $E$ nor $F$. In this case, $U_i$ is won by ${\cal M}_B$ for the same reason as in the preceding clause. If $U_i$ contains $E\wedge F$, then there are two undergroups $V_i^{E}$, $V_i^{F}$ of $A$ corresponding to $U_i$, where $V_i^{E}$ contains $E$ (but not $F$), and $V_i^{F}$ contains $F$ (but not $E$), with all other ($\neq E\wedge F$) oformulas of $U_i$ contained by both $V_i^{E}$ and $V_i^{F}$. Of course, both $V_i^E$ and $V_i^F$ are won by ${\cal M}_A$ because ${\cal M}_A$ wins the overall game $A$. This means that there is an oformula $G_1$ (resp. $G_2$) in $V_i^E$ (resp. $V_i^F$) such that ${\cal M}_A$ wins it. If at least one oformua $G\in\{G_1,G_2\}$ is neither $E$ nor $F$, then the corresponding oformula $G$ of $B$ is won by ${\cal M}_B$, because ${\cal M}_B$ plays in $G$ exactly as ${\cal M}_A$ does. Hence the $G$-containing undergroup $U_i$ of $B$ is won by ${\cal M}_B$. If $G_1$ is $E$ and $G_2$ is $F$, then ${\cal M}_A$ winning them means that ${\cal M}_B$ wins both the $E$ and the $F$ components of $E\wedge F$, because ${\cal M_B}$ plays in the $E$ (resp. $F$) component of $E\wedge F$ exactly as ${\cal M_A}$ does in $E$ (resp. $F$).  Hence $E\wedge F$ is won by ${\cal M}_B$, and hence so is the $E\wedge F$-containing undergroup $U_i$.\vspace{2mm}

{\bf (9)}\ Assume $B$ is obtained from $A$ by Recurrence Introduction. That is, the last ($a$'th) oformula of $B$ is $\cst F$, and the last ($a$'th) oformula of $A$ is $F$. We further assume that the number of overgroups of $B$ is $n$, and thus the number of overgroups of  $A$ is  $n+1$. In what follows,  $\vec{w}$ is any sequence of $n$ finite bitstrings, and $b$ is a positive integer not exceeding the number of oformulas of either cirquent.
If $b\neq a$, then ${\cal M}_B$ simply reinterprets every move $b;\vec{w},\epsilon.\alpha$ by either player in the imaginary play of $A$ as the move $b;\vec{w}.\alpha$ by the same player in the real play of $B$, and vice versa. If $b=a$, then ${\cal M}_B$ reinterprets, for any finite bitstring $u$, every move $a;\vec{w},u.\alpha$ by either player in the imaginary play of $A$ as the move $a;\vec{w}.u.\alpha$ by the same player in the real play of $B$, and vice versa.

Consider any undergroup $U_i^{B}$ of $B$. Let $\vec{x}=x_1,\ldots,x_n$ be any sequence of $n$ infinite but essentially finite bitstrings. ${\cal M}_A$'s winning $A$ means that $\Omega$ is a $\top$-won run of $A$ and that the corresponding undergroup $U_i^{A}$ of $A$ is won by ${\cal M}_A$. Then, for any sequence $\vec{y}=x_1,\ldots,x_n,x$, where $x$ is any infinite but essentially finite bitstring,  there is an oformula $F_b$ in $U_i^{A}$ such that $\Omega^{\preceq b;\vec{y}}$ is a $\top$-won run of $F_b$. If such $F_b$ is not the $a$'th oformula $F$, then, in the context of $\vec{x}$, the oformula $F_b$ of $B$ is also won by ${\cal M}_B$, i.e. $\Gamma^{\preceq b;\vec{x}}$ is a $\top$-won run of $F_b$, because ${\cal M}_B$ plays in the thread $\vec{x}$ of $F_b$ in $B$ exactly as ${\cal M}_A$ does in the thread $\vec{y}$ of $F_b$ in $A$. Hence $U_i^{B}$ is won by ${\cal M}_B$. If $F_b$ is the $a$'th oformula $F$, then, in the context of $\vec{x}$, the corresponding oformula $\cst F$ of $B$ is won by ${\cal M}_B$ as well, i.e. $\Gamma^{\preceq a;\vec{x}}$ is a $\top$-won run of $\cst F$. This is so because ${\cal M}_B$ plays in the thread $x$ of $F$ within the thread $\vec{x}$ of $\cst F$ exactly as ${\cal M}_A$ does in the thread $\vec{y}$ of $F$ in $A$. Namely, $(\Gamma^{\preceq a;\vec{x}})^{\preceq x}=\Omega^{\preceq a;\vec{y}}$. Since $\Omega^{\preceq a;\vec{y}}$ is a $\top$-won run of $F$, so is $(\Gamma^{\preceq a;\vec{x}})^{\preceq x}$. Further, due to the arbitrariness of $x$, $\Gamma^{\preceq a;\vec{x}}$ is a $\top$-won run of $\cst F$. Therefore, the $\cst F$-containing undergroup $U_i^{B}$ is won by ${\cal }M_B$.\vspace{2mm}

{\bf (10)}\ Finally, assume that $B$ is obtained from $A$ by Corecurrence Introduction. Let us assume that the last ($a$'th) oformula of $B$ is $\ccost F$, and the last ($a$'th) oformula of $A$ is $F$. And assume that $n$ $(n\geq 0)$ is the number of the {\em new} overgroups $U_j$ in which the $a$'th oformula $F$ was included when moving from $B$ to $A$. Let us further assume that all of such $n$ overgroups are at the end of the list of overgroups of either cirquent. In what follows, let $\vec{w}$ be any sequence of $m$ finite bitstrings, where $m$ is the total number of overgroups of either cirquent minus $n$. We construct the EPM ${\cal M}_B$ as follows.

If, when moving from $B$ to $A$, no new overgroups emerged to include the $a$'th oformula (i.e. $n=0$), then ${\cal M}_B$'s work is simple. What it should do is to ``synchronize" one single (fixed) thread of $F$ within each thread $\vec{x}$ of $\ccost F$ with the same thread $\vec{x}$ of $F$ in $A$.  Specifically, let $z$ be the infinite bitstring $000\ldots$ (note that it is essentially finite). ${\cal M}_B$ translates every move $a;\vec{w}.\alpha$ made by ${\cal M_A}$ in the imaginary play of $A$ as its own move $a;\vec{w}.u.\alpha$ in the real play of $B$, where $u$ is a finite initial segment of $z$ such that $u$ is not a proper prefix of any other finite bitstring $v$ already used in the real play within some move $a;\vec{w'}.v.\beta$. And whenever the environment makes a move $a;\vec{w}.v.\beta$ in the real play of $B$, if $v$ is a prefix of $z$, ${\cal M}_B$ translates it as the move $a;\vec{w}.\beta$ by the imaginary adversary of ${\cal M}_A$ in the play of $A$, otherwise (i.e. $v$ is not a prefix of $z$), ${\cal M}_B$ simply ignores it.

If, when moving from $B$ to $A$, the $a$'th oformula was included by some overgroups of $A$ (i.e. $n\geq 1$), then as always ${\cal M}_B$ works by simulating ${\cal M}_A$. To describe its work, we need to generalize the concepts of fusion and defusion from the case of $n=2$ to the case of $n\geq 1$.

Let $x_1,\ldots,x_n$ be any $n$---finite or infinite---bitstrings. A bitstring $z$ is a {\bf fusion} of $x_1,\ldots,x_n$ iff $z$ is a shortest bitstring such that, for any $i\in\{1,\ldots,n\}$ and any positive integer $j$ not exceeding the length of $x_i$, the following condition is satisfied: the $(jn-n+i)$'th bit of $z$ exists and it is the $j$'th bit of $x_i$. For instance, if $x_1=000$, $x_2=11$, and $x_3=001$, then the fusions of $x_1,x_2,x_3$ are $010010001$ and $010010011$. Note that when all $n$ bitstrings are infinite, they have a unique fusion, as before. The {\bf $n$-defusion} of a bitstring $z$ is the $n$-tuple $(x_1,\ldots,x_n)$, where each $x_i$ is the result of deleting from $z$ all bits except those that were found in positions $j$ such that $j$ modulo $n$ equals $i$. For instance, the $4$-defusion of $00110101101001111$ is $(00101,0101,1011,1101)$. It is obvious that the generalized concepts of fusion and defusion also preserve the essentially finiteness.

Now about the work of ${\cal M}_B$. Whenever the environment makes a move $a;\vec{w},\epsilon,\ldots,\epsilon.u.\alpha$ ($n$ occurrences of $\epsilon$ after $\vec{w}$) in the real play of $B$, ${\cal M}_B$ translates it as the move $a;\vec{w},u_1,\ldots,u_n.\alpha$ made by the imaginary adversary of ${\cal M}_A$ in the play of $A$, where $(u_1,\ldots,u_n)$ is the $n$-defusion of $u$. Next, whenever ${\cal M}_A$ makes a move $a;\vec{w},u_1,\ldots,u_n.\alpha$ in the imaginary play of $A$, ${\cal M}_B$ translates it as a series of its own moves $a;\vec{w},\epsilon,\ldots,\epsilon.v_1.\alpha,\ \ \ldots,\ \ a;\vec{w},\epsilon,\ldots,\epsilon.v_k.\alpha$ in the real play of $B$, where $v_1,\ldots,v_k$ are all the fusions of $u_1,\ldots,u_n$.

As usual, consider any undergroup $U_i^{B}$ of $B$, and let $\vec{x}=\vec{y},x_1,\ldots,x_n$ be any sequence of $(m+n)$ infinite but essentially finite bitstrings, where $\vec{y}$ is any sequence of $m$ infinite but essentially finite bitstrings. Then the corresponding undergroup $U_i^A$ of $A$ is won by ${\cal M}_A$, which, in turn, means that there is an oformula $F_b$ ($1\leq b\leq a$) in $U_i^{A}$ such that ${\cal M}_A$ wins it. If such $F_b$ is not the $a$'th oformula $F$, then the corresponding oformula $F_b$ of $B$ is also won by ${\cal M}_B$, because ${\cal M}_B$ plays in $F_b$ of $B$ exactly as ${\cal M}_A$ does in $F_b$ of $A$. Therefore, the $F_b$-containing undergroup $U_i^{B}$ is won by ${\cal M}_B$. If $F_b$ is the $a$'th oformula $F$, then the corresponding oformula $\ccost F$ of $B$ is won by ${\cal M}_B$ as well. This is so because ${\cal M}_B$ plays in at least one thread of $F$ within $\ccost F$ of $B$ exactly as ${\cal M}_A$ does in $F$ of $A$. Precisely, we have
$(\Gamma^{\preceq a;\vec{y},x_1,\ldots,x_n})^{\preceq x}=\Omega^{\preceq a;\vec{y},x_1,\ldots,x_n}$, where $x$ is the fusion of $(x_1,\ldots,x_n)$. Thus the $\ccost F$-containing undergroup $U_i^{B}$ is won by ${\cal M}_B$.
\end{proof}

\begin{theorem}\label{mainth2}
Every cirquent provable in {\bf CL15}$(\cst)$ is uniformly valid.

Furthermore, there is an effective procedure that takes an arbitrary {\bf CL15}$(\cst)$-proof of an arbitrary cirquent $C$ and constructs a uniform solution of $C$.
\end{theorem}

\begin{proof}
Immediately  from Theorem \ref{mainth1} by induction on the lengths of {\bf CL15}$(\cst)$-proofs.
\end{proof}

\begin{theorem}\label{mainth3}
For any formula $F$, if {\bf CL15}$(\cst)\vdash F$, then $F$ is uniformly valid.

Furthermore, there is an effective procedure which takes any {\bf CL15}$(\cst)$-proof of any formula $F$ and constructs a uniform solution of $F$.
\end{theorem}

\begin{proof}
Immediately from Theorem \ref{mainth2} and Lemma \ref{apr14b}.
\end{proof}

Below, a {\bf uniformly valid $(\neg,\wedge,\vee,\cst,\ccost)$-principle} means the result of replacing every occurrence of the operator $\cst$ (resp. $\ccost$) by the symbol $!$ (resp. $?$) in some uniformly valid $(\neg,\wedge,\vee,\cst,\ccost)$-formula. Similarly, a {\bf uniformly valid $(\neg,\wedge,\vee,\st,\cost)$-principle} means the result of replacing every occurrence of the operator $\st$ (resp. $\cost$) by the symbol $!$ (resp. $?$) in some uniformly valid $(\neg,\wedge,\vee,\st,\cost)$-formula. The reason for introducing these technical concepts is merely to make it possible to directly compare the otherwise syntactically nonidentical $(\neg,\wedge,\vee,\cst,\ccost)$-formulas with $(\neg,\wedge,\vee,\st,\cost)$-formulas.

\begin{theorem}\label{superset}
The set of uniformly valid $(\neg,\wedge,\vee,\cst,\ccost)$-principles is a proper superset of the set of uniformly valid $(\neg,\wedge,\vee,\st,\cost)$-principles.
\end{theorem}

\begin{proof} The fact that the set of uniformly valid $(\neg,\wedge,\vee,\cst,\ccost)$-principles is a {\em superset} of the set of uniformly valid  $(\neg,\wedge,\vee,\st,\cost)$-principles is immediate from Theorems \ref{niu} and \ref{mainth3}. Furthermore, the former set is in fact  a {\em proper} superset of the latter set because, as proven in \cite{Japsep}, the formula $P\wedge\cst(P\rightarrow P\wedge P)\wedge\cst(P\vee P\rightarrow P)\rightarrow\cst P$ is uniformly valid while its counterpart $P\wedge\st(P\rightarrow P\wedge P)\wedge\st(P\vee P\rightarrow P)\rightarrow\st P$ is not.
\end{proof}

\section{A further result}
To make our investigation of the relationship between $\cst$ and $\st$ more comprehensive, in this section we show that $\cst$ is strictly weaker than $\st$ (and thus $\ccost$ is strictly stronger than $\cost$) in the sense that the formula $\st P\rightarrow\cst P$ is uniformly valid while its converse $\cst P\rightarrow\st P$ is not. The first part of this statement is immediate from Lemma \ref{lemma1} of Section 3. So, we only need to prove the second part.

\begin{theorem}\label{notuni}
The formula $\cst P\rightarrow\st P$ is not uniformly valid.
\end{theorem}

\begin{proof}
Let ${\cal M}$ be an arbitrary EPM, i.e. strategy of the machine $(\top)$. Below we construct a counterstrategy ${\cal C}$ such that, when the environment $(\bot)$ follows it, ${\cal M}$ loses $\cst P\rightarrow\st P$ with $P$ interpreted as a certain enumeration game. Here, an {\bf enumeration game} (\cite{Japsep}) is a game where any natural number, identified with its decimal representation, is a legal move by either player at any time (and there are no other legal moves). It should be noted that, as shown in \cite{Japtam2}, every enumeration game is static, and hence is a legitimate value of an interpretation $^*$ on any atom. Hence, due to the arbitrariness of ${\cal M}$, $\cst P\rightarrow\st P$ (i.e. $\ccost\neg P\vee\st P$) is not uniformly valid.

Since $P$ is going to be interpreted as an enumeration game and its legal moves are known even before we actually define that interpretation, in certain contexts we may identify formulas with games without creating any confusion. The work of ${\cal C}$ consists in repeating the following interactive routine over and over again (infinitely many times), where $i$ is the number of the iteration. In our description below, a {\em fresh number} means a natural number that has not yet been chosen in the play by either player as a move in any thread/copy of $P$.\vspace{2mm}

LOOP($i$): Whenever permission is granted by the machine ${\cal M}$, make the move $2.w.u$, where $u$ is a fresh number and $w$ is the $i$th finite bitstring of the lexicographic list of all finite bitstrings.
\vspace{2mm}

Consider the run $\Delta$ generated by ${\cal M}$ in the scenario when its adversary follows the above counterstrategy. Let $\Omega=\Delta^{1.}$ and $\Gamma=\Delta^{2.}$. That is, $\Omega$ is the (sub)run that took place in the $\ccost\neg P$ component, and $\Gamma$ is the (sub)run that took place in the $\st P$ component. From some analysis of the work of LOOP, details of which are left to the reader, one can see that  $\Gamma^{\preceq x_1}\neq\Gamma^{\preceq x_2}$ for any two different infinite bitstrings $x_1$ and $x_2$. Hence, as there are uncountably many infinite bitstrings while only countably many infinite but essentially finite bitstrings, there is an infinite bitstring $y$ such that, for every infinite but essentially finite bitstring $v$,  $\Omega^{\preceq v}\neq\neg\Gamma^{\preceq y}$. Fix this $y$.


Now we select an interpretation $^*$ that interprets $P$ as the enumeration game such that, for any legal run $\Theta$ of the game $P$, $Wn^{P}\langle\Theta\rangle=\bot$ iff $\Theta=\Gamma^{\preceq y}$. We claim that ${\cal M}$ loses the overall game under this interpretation. First, it is obvious that ${\cal M}$ loses the game $P$ in the thread $y$, which means that it loses the $\st P$ component. Next, ${\cal M}$ also loses the $\ccost\neg P$ component because it loses in every essentially finite thread of $\neg P$ within $\ccost\neg P$. This is so because the run that took place in any essentially finite thread of $\neg P$ within $\ccost\neg P$ is won by $\top$ iff it is $\neg\Gamma^{\preceq y}$, which, however, is impossible (due to the above analysis).
\end{proof}

An alternative albeit non-constructive and less direct proof of Theorem \ref{notuni} would rely on Theorem \ref{superset}. Namely, one could show that, if $\cst P\rightarrow\st P$ was uniformly valid and hence (in view of the already proven fact of the uniform validity of the converse of this formula) $\st P$ and $\cst P$ were ``logically equivalent", then they would induce identical logics, in the precise sense that the set of uniformly valid $(\neg,\wedge,\vee,\cst,\ccost)$-principles would coincide with the set of uniformly valid $(\neg,\wedge,\vee,\st,\cost)$-principles, contrary to what Theorem \ref{superset} asserts.

\end{document}